\documentclass[journal,12pt,draftclsnofoot,onecolumn]{IEEEtran}
\IEEEoverridecommandlockouts
\usepackage{amsthm}
\usepackage{pifont}
\usepackage{bbding}
\usepackage{bbm}
\usepackage{latexsym}
\usepackage{epic}
\usepackage{mathrsfs}
\usepackage{amssymb,color,graphicx}
\usepackage{algorithm}
\usepackage{algorithmic}
\usepackage{graphicx}
\usepackage{subfigure}
\usepackage{cite}
\usepackage{float}
\usepackage{stfloats}
\usepackage[T1]{fontenc}
\usepackage{carolmin}
\usepackage{setspace}
\usepackage{makecell}
\usepackage{colortbl}
\usepackage{multirow,booktabs} 
\usepackage{amsmath,amssymb,amsfonts}
\theoremstyle{plain}
\usepackage[dvipsnames]{xcolor}
\usepackage{algorithmic}
\usepackage{ifpdf}

\newtheorem{Lemma}{Lemma}
\newtheorem{Theorem}{Theorem}
\newtheorem{Corollary}{Corollary}
\newtheorem{Definition}{Definition} 
\graphicspath{{./figures/}}

\begin{document}
\author{Cunyan Ma, Xiaoya Li, Chen He, \textit{Member, IEEE}, Jinye Peng, and Z. Jane Wang, \textit{Fellow, IEEE}}
 


\title{ Effects of 3D Position Fluctuations on Air-to-Ground mmWave UAV Communications
}
\maketitle
\begin{abstract} Millimeter wave (mmWave)-based  unmanned aerial vehicle (UAV) communication is a promising candidate for future communications due to its  flexibility and sufficient bandwidth.
However, random fluctuations in the position of hovering UAVs will lead to random variations in the  blockage  and signal-to-noise ratio (SNR) of  the UAV-user link, thus affecting   the quality
of service (QoS) of the system.
 To   assess the impact of UAV position fluctuations on the QoS of air-to-ground  mmWave UAV communications,  this paper develops a tractable analytical model that  jointly captures the features of three-dimensional (3D) position fluctuations of hovering UAVs and blockages of  mmWave (including static, dynamic, and self-blockages). 
   With this model, we derive the closed-form expressions for  reliable service probability respective to blockage probability of  UAV-user links,
    and  coverage probability respective to   SNR, respectively.
The results indicate that the greater the  position fluctuations of UAVs, the  lower the reliable service probability and  coverage probability.
 The degradation of these two  evaluation  metrics  confirms that the performance of air-to-ground mmWave UAV systems largely depends on the  UAV position fluctuations, and  the stronger the fluctuation,  the worse the QoS.
Finally, Monte Carlo simulations demonstrate the above results and show UAVs' optimal location to maximize the reliable service and coverage probability, respectively.

\end{abstract}
\IEEEpeerreviewmaketitle


\section{Introduction}
The advantage  of unmanned aerial vehicle (UAV)  makes it more likely to be employed as the aerial  base station (BS) for future communications, which will significantly  satisfy the diversified requirements on data rate, transmission delay, system capacity, etc.\cite{2022toward,2023A}.
Compared with traditional  BSs,  UAVs have higher flexibility and lower cost. 
Therefore, UAV communications have attracted increasing attention and are particularly suitable for on-demand scenarios, such as hot spots or disaster sites where existing communications infrastructures are overloaded or damaged
\cite{2023UAV-Correlated}.
Meanwhile,  to achieve ultra-high data rates, more and more researchers consider millimeter wave (mmWave) with large available bandwidth as the carrier frequency of UAV communications\cite{2020Clustered}. Such communications are hereinafter called as  mmWave UAV communications.   mmWave UAV communications have significant advantages  in flexibility, system capacity,  etc. 

One of the critical challenges of deploying mmWave UAVs is  that the position and orientation of hovering UAVs, compared to ground  BSs, can change randomly due to wind, inaccurate positioning,  
etc.\cite{2020Statistical,2022Millimeter,20213DChannel}. 
Such random  variations will lead to random changes in  the quality of the
 UAV-user links, which will lead to  unreliable  communications. Some innovative studies have investigated the impact of hovering  fluctuations of UAVs on mmWave UAV communication systems. For example,
the work in \cite{2022Millimeter} experimentally investigated the effect of random fluctuations of hovering UAVs
  on mmWave signals. It demonstrated that the mismatch of antenna beams caused by the random fluctuations  would deteriorate the performance of UAV-to-ground  links.
The authors in \cite{20213DChannel,2020Analytical} investigated the influence of UAV orientation fluctuations on system outage performance. They showed that the beam misalignment and antenna gain mismatch caused by UAV orientation fluctuations lead to a decline in the reliability of the  communication system.   The authors in
 \cite{wang2021jittering} analyzed the effect of UAV attitude angle fluctuations on the antenna orientation, and they proposed a  UAV beam training method to improve the accuracy of antenna angle estimation.
However, \cite{2022Millimeter,20213DChannel, 2020Analytical,wang2021jittering} focused on investigating the impact of antenna beam misalignment or poor beam selection caused by hovering UAV fluctuations on system performances. 

In addition to the above impacts on antenna beam alignment,   fluctuations, especially position fluctuations of UAVs may  also significantly impact mmWave blockage characteristics. 
As we know, mmWave  signals are easily blocked by obstacles, and 
the blockage state is mainly determined by the relative positions of transceivers and obstacles \cite{2020Efficient}. 
The random position deviation  of hovering UAVs will make the blockage characteristics more randomized, thus affecting the  system's performance.
So investigating the influence of UAV position fluctuations on  link blockage,  then evaluating the   QoS is  necessary.
Specifically, UAV position fluctuation 
is three-dimensional (3D), 
and a line-of-sight (LoS) mmWave link may encounter three types of  blockages:  static blockage caused by static obstacles such as  buildings, dynamic blockage caused by  moving blockers, and self-blockage caused by the user's own body\cite{2019The}. 
Therefore,   it is essential to obtain the relationship between  the 3D position fluctuations of hovering UAVs and the three kinds of blockages of mmWave, as well as the impact on the system performance.

The effect of  blockages on the mmWave system without considering the position fluctuations of UAVs has been widely studied.
 For instance, the work in  \cite{2021THE, 2019The,2018Driven, 2017On,2016Analysis,2014Analysis} studied  the blockage effects of mmWave signals in  terrestrial communications.  
 For mmWave UAV communications,  the static blockage due to buildings has been studied in \cite{2022Geometric,2021Line,2014Optimal}, and the LoS probability is obtained to guide the deployment of UAVs. 
In \cite{2020Clustered, guo2022coverage, shi2022modeling}, the coverage performance  was analyzed under static blockage.
Furthermore, the authors in \cite{2018Flexible} contributed an  analytical framework to characterize mmWave backhaul links by considering dynamic   blockage, and it is  shown that  the assistance of UAVs can significantly improve the performance of mmWave backhaul links. 
The authors  in \cite{2018Effects} presented a method of optimal deployment, which provided the optimal 3D position and coverage radius of a    UAV  by considering human  blockage. 
More recently, in \cite{2022Characterization}, the authors  experimentally verified the impact of  human body  on  air-to-ground  links, and the results showed that the  effect  depended largely on the hovering UAVs' positions.
 However, the above results 
  are obtained under the supposition  that the  position of the UAV  is perfectly stable. This assumption is too ideal because, in practice, the fluctuation of UAVs is unavoidable.
 
As far as we know, there are few studies on  mmWave blockage characteristics under the effect of UAV position fluctuations. 
In \cite{2020Robust}, the authors studied the blockage characteristics and beam misalignment considering the hovering fluctuation effect of UAV, and they proposed a UAV deployment and beamforming optimization method to establish robust transmission. However, their theoretical model only analyzed the blockage probability considering UAV position fluctuations. 
But, more importantly,  further system performance analysis still needs improvement. Especially the system reliability analysis under the joint consideration of fluctuations and blockages has yet to be carried out.
In addition, in our previous work\cite{2022Effects}, we developed a theoretical analysis model to study the impact of UAV position fluctuations on UAV-user link blockages and the reliability of air-to-ground mmWave UAV communications. 
However, the results of  \cite{2022Effects} are  obtained based on a simple model that only considered the height fluctuations of UAVs.

Although the above work provides valuable insights for mmWave UAV communications, more research still needs to be done on the effect of  UAV 3D position fluctuations on the system's QoS.
A comprehensive analytical model that builds the relationship between the hovering UAVs' position fluctuations and the mmWave link's blockages is greatly requested
 to benefit from mmWave UAV communications.
This is what we are trying to do through this work. 
The major contributions of this paper are given as follows:

$\bullet$ We propose a novel  model to analyze the effect of UAV 3D position fluctuations on the QoS of
air-to-ground mmWave UAV  systems.
Different from existing works, this work establishes the theoretical relationship between the  position fluctuations of UAVs   and  the blockages of 
mmWave links
(including static, dynamic, and self-blockages), which enables the evaluation of the blockage characteristics of UAV-user links under  position fluctuations.
Especially,  we find that the dynamic blockage   is obviously affected by the random fluctuations of UAV position, and its fluctuation variance is directly 
 proportional to the  variance of  UAV position fluctuations.

$\bullet$ We derive the  closed-form expressions for  reliable service probability  respective to  the probability of UAV-user links blockages,
and coverage probability respective to the signal-to-noise ratio (SNR), respectively. 
The results indicate that the greater the position fluctuation  strength of the UAV is, the lower the reliable service probability and coverage probability is. Therefore,  UAV
position fluctuations considerably influence these two QoS metrics of the considered system. To our best knowledge, this is the first work that analytically evaluates the impact of UAV 3D position fluctuations on the QoS of air-to-ground mmWave UAV communication
systems while considering the static, dynamic, and self-blockage of mmWave links.

$\bullet$ We provide Monte Carlo simulations to verify the  theoretical results, and the simulation results also show that there exists an  optimal horizontal position and height of UAVs  to maximize the reliable service probability and   coverage probability, respectively, which  help  establish  reliable air-to-ground mmWave UAV communications.

The rest of this paper is organized as follows: In Section II,  we present
the system model. In Section III, we analyze the effect of UAV position fluctuation on link blockages.
 In Section IV,
the closed-form expression of reliable service probability
  is derived, while in Section V,  the closed-form expression of coverage probability is derived. Section VI provides the
 simulation results, and Section VII concludes  this paper.  Notations used in this paper are given in Table I.

\section{System Model}
We consider a mmWave UAV communications scenario, as shown in Fig. \ref{fig:system}, where ground users are randomly distributed  and served by  UAVs as aerial BSs. Since obstacles easily block UAV-user links, we assume the user can immediately link  to another unblocked link once the current UAV-user link is blocked.
 Without losing generality, we randomly select a user as the typical user.  The typical user's location is denoted as $(x, y, z) \!=\! (0, 0, h_R)$, where $h_R$ is the height of the typical user.
The specific  model settings are given as follows:

\begin{table} \label{notation}
\renewcommand\arraystretch{1}
\caption{\rule{0pt}{12pt}Description of Notation}
\begin{center}
\begin{tabular}{|l|l|}  
\hline \textbf{\;\;\;\;Notation}               &\textbf{\;\;\;\;\;\;\;\;\;\;\;\;\;\;\;\;\;\;\;\;Description}   \\
\hline   $(x_i, y_i, h_i)$ & Real-time position  of the $i$-th UAV under  fluctuations. \\
 $(\mu_{x_i}, \mu_{y_i}, \mu_{h_i})$  & Mean of the  real-time position  $(x_i, y_i, h_i)$.\\
$\lambda_T$/$\lambda_B$/$\lambda_S$  &Density of UAVs/human blockers/buildings.\\
  $v$         &Velocity of moving human blockers. \\
  $r_i$  &2D distance from the $i$-th UAV to the user.\\
$B_i^{\rm{d}}$/$B_i^{\rm{s}}$            &Indicator for dynamic/static blockage   for the $i$-th  link.  \\
  $B^{\rm{sel}}$        &Indicator for self-blockage  for a single  UAV-user link.\\

$C_i$           &Indicator for the $i$-th UAV is available.\\
 
 $M$/$N$  &Number of   all/available UAVs.\\


 \hline
\end{tabular}\label{table:SFDs}
\end{center}
\end{table}

\begin{figure}
\centering
  \includegraphics[scale=0.4]{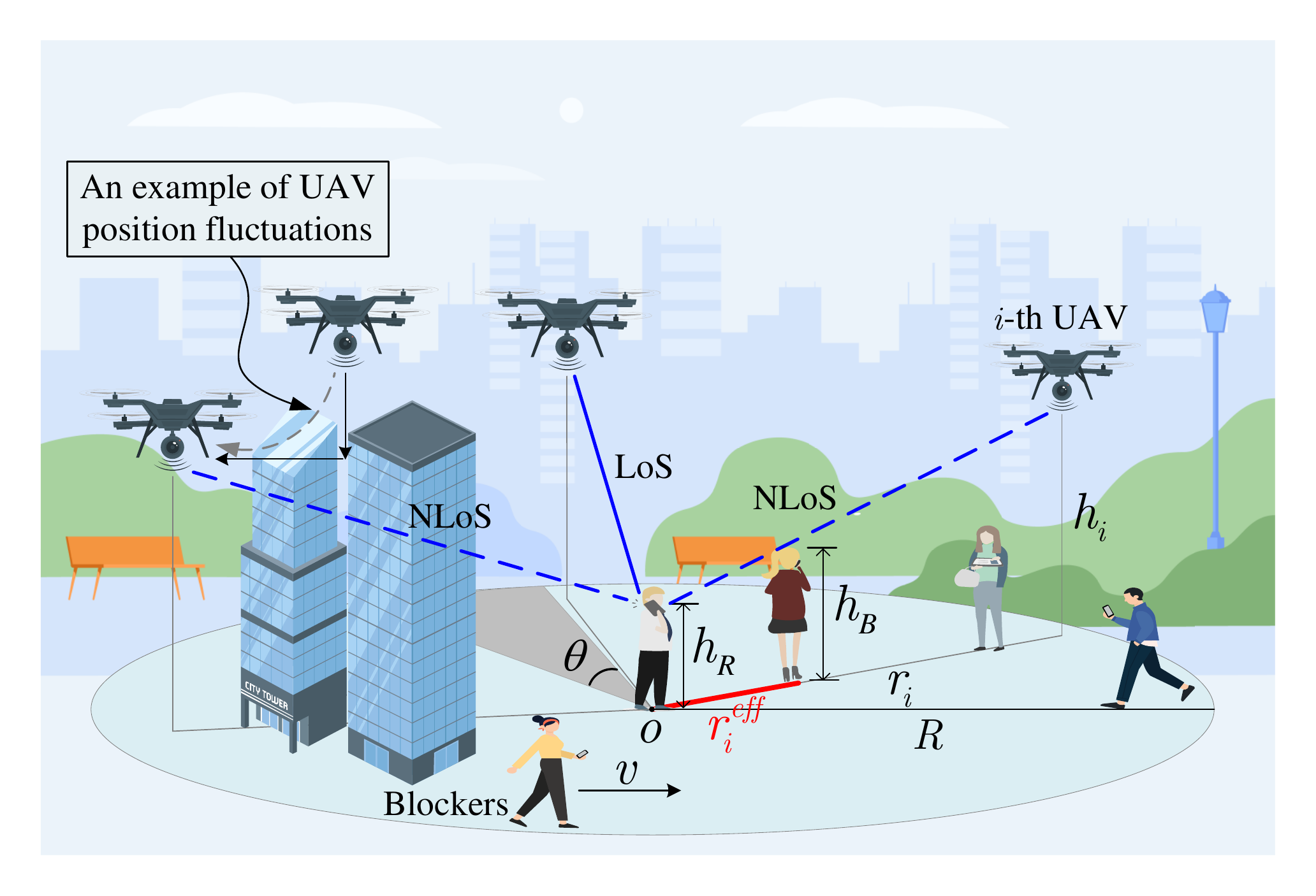} \\  
  \caption{The considered mmWave UAV communications, where the position of  the hovering UAVs will fluctuate randomly,
the QoS provided to the user will be affected by various blockages and fluctuations.
}\label{fig:system}
\end{figure}

\subsection{UAV Models}  \label{UAVs}
This  subsection models the position and position fluctuations of   UAVs. Since   hovering UAVs will fluctuate randomly, the real-time position of each UAV  fluctuates randomly around its central position \cite{2020Statistical,2020Robust}. We denote $(\mu_{x_i}, \mu_{y_i}, \mu_{h_i})$
and $({x_i}, {y_i}, {h_i})$ as the central position    and the  real-time position for the 
$i$-th UAV, respectively.  The central positions 
on the two-dimensional (2D) plane is modeled as a homogeneous Poisson point process (PPP) with   density $\lambda_T$ \cite{2019The}. We assume that 
the maximum service radius of the UAV is $R$. Therefore, the number of all potential serving UAVs $M$ for the typical user is Poisson distributed 
i.e., $ \mathbb{P}_M(m)=\frac{{[ \lambda_T\pi{R}^{2}]}^{m}}{m!}\exp\left({\!-\!\lambda_T\pi{R}^{2}}\right)$.
 Then, due to
 the  real-time position of the hovering  UAV fluctuates randomly,
  and numerous literature  has modeled  UAV  position  fluctuations as  Gaussian distributions \cite{2021Hovering,20213DChannel,2020Statistical,2018Channel},
 the exact position of the $i$-th UAV 
 can be modeled as  $P_i\sim N (\mu_{P_i}, \sigma^2), P\in\{x,y,h\} $, 
 where 
 $\sigma^2$ is the variance of the UAV position fluctuations. 
It is clear that the greater the  $\sigma$ is, 
the greater the   fluctuation  is.
 Therefore, we regard $\sigma$ as a measure of  UAV position fluctuation strength.

\subsection{Blockage Models}  \label{Blockage Models}
\subsubsection{Dynamic Blockage Model}
Dynamic blockage due to moving humans has been  thoroughly studied in mmWave systems
with traditional BSs \cite{2021THE,2019The,2018Driven}.
Following \cite{2019The}, we assume that  humans are located in the  region  according to a homogeneous PPP of  density $\lambda_B$, and they  move randomly at  velocity $v$.  
 We define effective blockage length to reflect whether
blockage occurs when a human blocker passes through a UAV-user link.
The  effective blockage length $r_i^{\mathrm{eff}}$ for the $i$-th UAV-user  link  is  $ r_i^{\mathrm{eff}}\!=\!\frac{h_B-h_R}{h_i-h_R}r_i$,
where $h_B$ is the height of the human  blocker,  $r_i\!=\!\sqrt{x_i^2+y_i^2}$. 
Then, the dynamic blockage of the $i$-th   link is modeled as  an exponential on-off process with blocking and non-blocking rates of $\eta_i$ and $\omega$, respectively \cite{2019The}, where  
$ \eta_i=\frac{2}{\pi}\lambda_B v r_i^{\mathrm{eff}}$ and $\omega$ is simplified to  a constant.  Moreover,  we know that the $i$-th link is dynamically blocked with probability (w.p.) $\phi_i(h_i,r_i)=\frac{\eta_i}{\eta_i+\omega}$. We define a random variable $B_i^{\rm{d}}\!=\!\{1,0\}$ to indicate whether  dynamic blockage occurs on the $i$-th UAV-user link,  and $B_i^{\rm{d}}$ can be represented as
\begin{align}  \label{eq:P_B''}
   B_i^{\rm{d}}\!=\!\left\{
    \begin{array}{l l l}
     1
  &{\text{w.p.}}  &\phi_i(h_i,r_i)=\frac{\rho{r_i}}{{\rho r_i+\omega(h_i-h_R)}}, \\ 
       0
   &{\text{w.p.}}     &\tilde{\phi_i}(h_i,r_i)=1-\phi_i(h_i,r_i), \\  
  \end{array} \right.
\end{align}
where $\rho={2}{\lambda_B}v({h_B-h_R})/{\pi}$ is used for convenience.

\subsubsection{Static Blockage Model}
Static blockage caused  by static obstacles (e.g., high-rise buildings) has been well studied in  \cite{2021Line,2020Clustered,2014Analysis} based on random shape theory. 
We consider the blockage model given
in \cite{2014Analysis} to model the static blockage, and the $i$-th link is blocked with probability $\psi_i(h_i,r_i)$,  where $\psi_i(h_i,r_i)$
 will be given later. 
 Similar to dynamic blockage, 
we define $B_i^{\rm{s}}\!=\!\{1,0\}$ to describe whether  static blockage occurs on the $i$-th  link,  and $B_i^{\rm{s}}$ can be represented as
\begin{align}  \label{eq:static1}
   B_i^{\rm{s}}\!=\!\left\{
    \begin{array}{l l l}
     1
  &{\text{w.p.}}  &\psi_i(h_i,r_i)=1-{\exp{\left(-(\epsilon r_i+\epsilon_0)\right)}}, \\ 
       0
   &{\text{w.p.}}     &\tilde{\psi_i}(h_i,r_i)=1-\psi_i(h_i,r_i), \\  
  \end{array} \right.
\end{align}
where  $\epsilon_0=\lambda_S\mathbb{E}(l)\mathbb{E}(w)$ and $\epsilon=\frac{2} {\pi}\lambda_S(\mathbb{E}(l)+\mathbb{E}(w))$,   $\lambda_S$, $\mathbb{E}(l)$, and $\mathbb{E}(w)$ are the density, expected length, and width of the buildings, respectively.  

\subsubsection{Self-Blockage Model}
Except for  dynamic and static blockages, a fraction  of   links will also be blocked by the user's own body.  
Following \cite{2019The}, we define a sector of angle $\theta$  behind the  user as the self-blockage zone.
The user himself will  block  all  UAV-user links in this zone. Therefore, the self-blockage
probability
 of a UAV-user link is  the probability that the UAV is located in the self-blockage zone.
We define an indicative random variable $B^{\rm{sel}}\!=\!\{1,0\}$ to describe whether  
self-blockage occurs on a single UAV-user link,  and $B^{\rm{sel}}$ can be represented as
\begin{align}  \label{eq:self}
   B^{\rm{sel}}\!=\!\left\{
    \begin{array}{l l l}
     1
  &{\text{w.p.}}  &\frac{\theta}{2\pi}, \\ 
       0
   &{\text{w.p.}}     &1-\frac{\theta}{2\pi}. \\  
  \end{array} \right.
\end{align}

\section{Effect of Fluctuations  on Blockages}
This section focuses on obtaining 
 the relationship between UAV position fluctuations and UAV-user link blockages.
Specifically, according to \eqref{eq:P_B''} and \eqref{eq:static1},  we observe that    the probability of dynamic and static blockages of UAV-user link is related to the position of the UAV (i.e., $h_i$ and $r_i$).  Due to random fluctuations in the UAV's position,  $\phi_i(h_i,r_i)$ and $\psi_i(h_i,r_i)$
 are  randomly changing  and  only reflect an instantaneous moment's blockage characteristic. 
On the one hand, the random variation in the blockage state is detrimental to users,
as it will lead to unreliable communications.
On the other hand, analyzing the performance of a  system at a specific moment is difficult and of limited value.  Hence, to evaluate the average performance of the   system, it is important to  obtain the distribution of  $\phi_i(h_i,r_i)$ and $\psi_i(h_i,r_i)$ under UAV position fluctuations. 
Accordingly, in the following, we derive the relevant probability density function (PDF).

 \subsection{Effect of Fluctuations  on Dynamic Blockage}
\begin{Theorem} \label{Theorem1}
The PDF of   $\phi_i(h_i,r_i)$ can be well approximately modeled as 
\begin{align}  \label{eq:PDF_PB}
    f_{\phi_i}(\phi_i)\approx\frac{1}{\sqrt{2\pi}\sigma_{\phi_i}}{\exp\left({-{\frac{{(\phi_i-\mu_{\phi_i})}^{2}}{2{\sigma_{\phi_i}^2}}}}\right)},
\end{align}
where the mean $\mu_{\phi_i}$ and variance $\sigma^2_{\phi_i}$
are given in  \eqref{eq:P_B} and \eqref{eq:P(B|C)}, respectively. Note that the notation ${\phi_i}$ is a compressed version of $\phi_i(h_i,r_i)$, which we use for convenience.
\end{Theorem}
\begin{proof}
The proof is given in Appendix A.
\end{proof}

\textit{ Remark 1: From Theorem \ref{Theorem1}, \eqref{eq:P_B} and \eqref{eq:P(B|C)},
we  can observe that the mean $\mu_{\phi_i}$ of $\phi_i(h_i,r_i)$ is independent of the variance $\sigma^2$ of the UAV position fluctuations. However, the variance  $\sigma^2_{\phi_i}$ of $\phi_i(h_i,r_i)$ is an increasing function of $\sigma^2$.
Therefore, the higher the  fluctuation strength of the UAV, the larger the deviation of  $\phi_i(h_i,r_i)$ from  its mean $\mu_{\phi_i}$.}

Then, 
 the relationship between $\mu_{\phi_i}$,   $\sigma_{\phi_i}$ and  average 2D distance $\mu_{r_i}$ 
   is given as follows:
 \begin{Corollary}  \label{Corollary4}
$\mu_{\phi_i}$ and  $\sigma_{\phi_i}$ are increasing and decreasing functions of $\mu_{r_i}$, respectively, $i\!=\!1,\cdots,m$. 
Therefore, 
 the maximum values of $\mu_{\phi_i}$ can be expressed as
\begin{align}  \label{eq:maxP}   \mu^{\max}_{\phi_i}=\frac{\rho{\mu^{\max}_{r_i}}}{{\rho
    {\mu^{\max}_{r_i}}+\omega(\mu_{h_i}\!-\!h_R)}}  
\end{align}
where $\mu^{\max}_{r_i}$  
denotes the maximum  value of $\mu_{r_i}$,   $i=1,\cdots,m$. Similarly, we can obtain the maximum  value of $\sigma_{\phi_i}$  by replacing $\mu_{r_i}$ in \eqref{eq:P(B|C)} with $\mu^{\min}_{r_i}$.
\end{Corollary}  
\begin{proof}
The proof is given in Appendix B.
\end{proof}
\textit{ Remark 2: According to Corollary \ref{Corollary4}, we  can conclude that  
 if the  UAV is deployed closer to the user, i.e., $\mu_{r_i}$ is smaller,
 the mean $\mu_{\phi_i}$ of  $\phi_i(h_i,r_i)$  will be smaller.
However, in this case, the  variance $\sigma_{\phi_i}^2$ of  $\phi_i(h_i,r_i)$  will become larger. 
Therefore, there may be an appropriate $\mu_{r_i}$ to obtain a trade-off  between  $\mu_{r_i}$ and $\sigma_{\phi_i}^2$ so that the  system can achieve a better QoS.}

 \subsection{Effect of Fluctuations  on Static Blockage}
\begin{Corollary}  \label{Corollary_s}
The PDF of   $\psi_i(h_i,r_i)$ can be well approximately modeled as 
\begin{align}  \label{eq:log-normally}
  f_{\psi_i}(\psi_i)\approx \frac{1}{\sqrt{2\pi}{(1\!-\!\psi_i)}\sigma_{\hat{r_i}}}{\exp\left(-{{\frac{({\ln\left({1-\psi_i}\right)-\mu_{\hat{r_i}})}^{2}}{2{\sigma_{\hat{r_i}}^2}}}}\right)}, 
\end{align}
where   $\sigma_{\hat{r_i}}$ and $\mu_{\hat{r_i}}$ can be seen in the proof, and ${\psi_i}$ is a compressed version of $\psi_i(h_i,r_i)$.

\end{Corollary}
\begin{proof}
We first derive the PDF of $\tilde{\psi_i}(h_i,r_i)={\exp{\left(-(\epsilon r_i+\epsilon_0)\right)}}$.
According to  \eqref{eq:Gaussian_r}, $r_i$ can be approximately modeled as a normal distribution  with mean $\mu_{r_i}$ and
variance $ \sigma^2$. As a result, we have
$-(\epsilon r_i+\epsilon_0)$  is  normally distributed  with mean
$\mu_{\hat{r_i}}\!=\!-(\epsilon \mu_{r_i}+\epsilon_0)$ and variance  ${\sigma_{\hat{r_i}}^2}\!=\!\epsilon^2\sigma^2$. Therefore,   $\tilde{\psi_i}(h_i,r_i)$  is  log-normally distributed \cite{1993Information}, and its  PDF is given by
$f_{{\tilde{\psi_i}}}(\tilde{\psi_i})\approx \frac{1}{\sqrt{2\pi}{\tilde{\psi_i}}\sigma_{\hat{r_i}}}{\exp\left({-{\frac{({\ln\left({\tilde{\psi_i}}\right)-\mu_{\hat{r_i}})}^{2}}{2{\sigma_{\hat{r_i}}^2}}}}\right)}$.
We further denote the cumulative distribution function (CDF) of $\tilde{\psi_i}(h_i,r_i)$  
and ${\psi_i(h_i,r_i)}$ as $F_{\tilde{\psi_i}}(\tilde{\psi_i})$
and $F_{\psi_i}({\psi_i})$, respectively. It is easy to get 
$F_{\psi_i}({\psi_i})\!=\!\mathbb{P}(1\!-\!{\exp{\left(-(\epsilon r_i+\epsilon_0)\right)}}\le{\psi_i)}=1-F_{\tilde{\psi_i}}(1-{\psi_i})$, and 
$f_{{{\psi_i}}}({\psi_i})=f_{{\tilde{\psi_i}}}(1-{\psi_i})$ by
taking the derivative of  $F_{{{\psi_i}}}({\psi_i})$.
Finally,
 the PDF of ${\psi_i(h_i,r_i)}$ is obtained  as shown in \eqref{eq:log-normally}.
\end{proof}
Using Corollary \ref{Corollary_s}, we can obtain the 
 mean 
 and variance  of $\tilde{\psi_i}(h_i,r_i)$ separately as
 \begin{align} 
    \mu_{\tilde{\psi_i}}&={\exp{\left(\mu_{\hat{r_i}}+{\sigma_{\hat{r_i}}^2}/2\right)}}, \label{eq:log-normally1} \\
    \sigma_{\tilde{\psi_i}}^2&={\exp{\left(2(\mu_{\hat{r_i}}+\sigma_{\hat{r_i}}^2)\right)}}-{\exp{\left(2(\mu_{\hat{r_i}}+\sigma_{\hat{r_i}}^2/2)\right)}} \label{eq:log-normally2}
\end{align}
\textit{Remark 3: From \eqref{eq:log-normally1} and \eqref{eq:log-normally2}, we can infer that the mean and variance of  $\psi_i(h_i,r_i)$  may also be affected by the position fluctuation of UAV. However, in our system assumption, we can obtain:
$\mu_{\tilde{\psi_i}}\!\approx\!\exp{\left(-(\epsilon \mu_{r_i}\!+\!\epsilon_0\right))}$ and $\sigma_{\tilde{\psi_i}}^2\approx0$, where the detailed proof can be seen in Appendix \ref{Appendix C}. 
The results  indicate that the 
$\psi_i(h_i,r_i)$ hardly changes with  the position fluctuation of the UAV, which matches the intuition that the dimension of the position fluctuation of the UAV is very small compared to the size of  buildings.
 Therefore, the impact of UAV position fluctuations on the static blockage probability of the link is negligible, and only when the UAV is located at the edge of the building will position  fluctuations of the UAV affect the link' static blockage state.}


\section{Effect of Fluctuations on Reliable Service}  \label{Blockage}
 This section analyzes the effect  of UAV position fluctuations on the reliable service probability. 
 Since static  and self-blockage will lead to permanent blockage,  a UAV is available if  UAV-user link  is not blocked by   static blockage and self-blockage.  
We are more concerned  that at least one UAV is available.   Due to  blockages  affecting the system reliability, we propose a new QoS metric in terms of links blockages,  called reliable service probability, given by Definition 1.
 \begin{Definition}
A user is said to be in reliable service  if  at least one UAV is available with a blockage probability  not higher than a predefined threshold.
The reliable service probability is denoted as $\mathbb{P}_{\mathrm {rel}}$ and
given by 
 \begin{align} \label{eq:rel}
 \mathbb{P}_{\mathrm {rel}}\triangleq1\!-\!\sum_{n=0}^{\infty} \mathbb{P}_N(n)\prod_{i=0}^{n}\mathbb{P} \left(\phi_i(h_i,r_i)\!>\! p_{\mathrm{th}}\right),
 \end{align}
where $\mathbb{P}_N(n)$  denotes  the probability that $n$ UAVs are available, 
$\mathbb{P}_N(n)\prod_{i=0}^{n}\mathbb{P} \left(\phi_i(h_i,r_i)\!>\! p_{\mathrm{th}}\right)$  denotes the probability that $n$ UAVs are available and the dynamic blockage probability of each link is greater than the  threshold  $p_{\mathrm{th}}$.
Therefore, the second term of \eqref{eq:rel} represents the probability that the user is not in a reliable service. 
Note that $n=0$  will result in $i=0$,
 which means  there are no UAVs available. So we let $\mathbb{P} \left(\phi_i(h_i,r_i)\!>\! p_{\mathrm{th}}\right)|_{i=0}=1$.
To calculate $ \mathbb{P}_{\mathrm {rel}}$, we first give the distribution of UAVs availability in the following.
  \end{Definition}
  
\subsection{Distribution of UAVs Availability}
 We use an indicative random variable  $C_i$ to indicate that  the $i$-th  UAV is available,   $ \mathbb{P}(C_i)$ denotes  the probability of  $C_i$, which is also the probability that the $i$-th  link 
 is not blocked by  static  and self-blockage.
 Then,    
the distribution of the number 
of available UAVs  is obtained as:
\begin{Lemma}  \label{Theorem_avai}
The distribution of the number 
of available UAVs $N$ is Poisson distributed
 with parameter $\mathbb{P}(C_i)\lambda_T\pi{R}^{2}$, i.e.,  
\begin{align} \label{eq:{P}_{N}}
   \mathbb{P}_{N}(n)=\frac{{[ \mathbb{P}(C_i)\lambda_T\pi{R}^{2}]}^{n}}{n!}\exp\left({-\mathbb{P}(C_i)\lambda_T\pi{R}^{2}}\right),
 \end{align}
  where  $\mathbb{P}(C_i)\approx\left(1\!-\!\frac{\theta}{2\pi}\right)\frac{2\exp{(\!-\epsilon_0)}}{R^2\epsilon^2}\left(1\!-\!(1\!+\!R\epsilon )\exp{(\!-R\epsilon )}\right)$.
\end{Lemma}   
\begin{proof}
The proof is given in Appendix C.
\end{proof}

Then,  we derive the closed-form expression of reliable service probability for single and multiple UAV cases according to  Definition 1.
It indicates  that the greater the  fluctuation strength $\sigma$ of the UAV, the lower the $\mathbb{P}_{\mathrm {rel}}$, and the worse the QoS. The specific details are given as follows.

\subsection{Reliable Service Probability for Single UAV Case}
In this case, we assume that the typical user can only link to one UAV, and denote the UAV as the $i$-th UAV.  
 Let $\mathbb{P}_{\mathrm {rel}}^{\mathrm{sig}}$  represent the reliable service probability for the single UAV case, it indicates the probability that the $i$-th UAV is available and the dynamic blockage probability of the   link is lower than the   threshold $p_{\mathrm{th}}$. 
 Then,
$\mathbb{P}_{\mathrm {rel}}^{\mathrm{sig}}$  can be  obtained as shown in Theorem \ref{Theorem_single}.
\begin{Theorem}  \label{Theorem_single}
The  reliable service probability  for the single UAV case is given by 
\begin{align}   \label{eq:cov0}
    \mathbb{P}_{\mathrm {rel}}^{\mathrm{sig}}\approx \mathbb{P}(C_i) \left(\frac{1}{2}+\frac{1}{2}\mathrm{erf}\left(\frac{p_{\mathrm{th}}-\mu_{\phi_i}}{\sqrt{2}\sigma_{\phi_i}}\right)\right),
\end{align}
 where  $\mathbb{P}(C_i)$ is given in \eqref{eq:C_i_appro}, $\mu_{\phi_i}$ and  $\sigma_{\phi_i}$ are given in  \eqref{eq:P_B} and \eqref{eq:P(B|C)}, respectively.
\end{Theorem}
\begin{proof}
Since  in this case we only consider one   UAV is available,    \eqref{eq:rel} can be rewritten as   
\begin{align}  \label{eq:P_single}   
    \mathbb{P}_{\mathrm {rel}}^{\mathrm{sig}}=1\!-\!\sum_{n=0}^{1} \mathbb{P}_N(n)\prod_{i=0}^{n}\mathbb{P} \left(\phi_i(h_i,r_i)\!>\! p_{\mathrm{th}}\right),
\end{align}
where $\mathbb{P}_N(1)$ is indeed the probability that the $i$-th  UAV is 
 available, i.e., $\mathbb{P}_N(1)=\mathbb{P}(C_i)$, and
$\mathbb{P}_N(0)=1-\mathbb{P}(C_i)$. 
Substituting $\mathbb{P}_N(0)$ and $\mathbb{P}_N(1)$
  into \eqref{eq:P_single}, then using Theorem \ref{Theorem1},
 we can get
\begin{align}  \label{eq:cov00}
   \mathbb{P}_{\mathrm {rel}}^{\mathrm{sig}}&=\mathbb{P}(C_i)- \mathbb{P}(C_i)\mathbb{P} \left(\phi_i(h_i,r_i)\!>\! p_{\mathrm{th}}\right)\nonumber \\
    &= \mathbb{P}(C_i)\int_{0}^{p_{\mathrm{th}}}  f_{\phi_i}(\phi_i)d{\phi_i} \nonumber \\
    &\approx \mathbb{P}(C_i) \left(\frac{1}{2}+\frac{1}{2}\mathrm{erf}\left(\frac{p_{\mathrm{th}}-\mu_{\phi_i}}{\sqrt{2}\sigma_{\phi_i}}\right)\right),
\end{align}
where  the last step  is obtained by integrating 
$f_{\phi_i}(\phi_i)$ according to 
\cite [Eq. (3)] {2016Analysis}. 
\end{proof}

According to Theorem \ref{Theorem_single}, we can obtain the effect of  $\sigma$ on the QoS. 
In \eqref{eq:cov0}, $\mathbb{P}(C_i)$ and $\mu_{\phi_i}$ are not affected by $\sigma$ according to \eqref{eq:C_i_appro} and \eqref{eq:P_B}, respectively.
So only $\sigma_{\phi_i}$ is related to $\sigma$, and $\sigma_{\phi_i}$ is an increasing function of $\sigma$, which can be observed from \eqref{eq:P(B|C)},  we first analyze the effect of $\sigma_{\phi_i}$ on $\mathbb{P}_{\mathrm {rel}}^{\mathrm{sig}}$.
Since  error function erf$(\cdot)$ is  a monotonically    function and erf$(0)\!=\!0$,  $\sigma_{\phi_i}$ will 
 show two 
 different effects on  $\mathbb{P}_{\mathrm {rel}}^{\mathrm{sig}}$. For  $p_{\mathrm{th}}> \mu_{\phi_i}$, 
erf$(\cdot)>0$  and  $\mathbb{P}_{\mathrm {rel}}^{\mathrm{sig}}$ decreases as $\sigma_{\phi_i}$ increases.
For  $p_{\mathrm{th}}< \mu_{\phi_i}$,  erf$(\cdot)\!<\!0$  and $\mathbb{P}_{\mathrm {rel}}^{\mathrm{sig}}$ increases as $\sigma_{\phi_i}$ increases. However, in the latter situation, the value of $\mathbb{P}_{\mathrm {rel}}^{\mathrm{sig}}$ is  lower, even  less than 0.5, which  indicates  that  the QoS  is evil, further analyze the effect of $\sigma$ on the QoS is almost meaningless.
Hence, throughout this paper,
 we pay more attention to the former case ($p_{\mathrm{th}}> \mu_{\phi_i}$). We can conclude that the greater the $\sigma$ is, the larger the $\sigma_{\phi_i}$ is,  the lower the $\mathbb{P}_{\mathrm {rel}}^{\mathrm{sig}}$ is, and the worse the QoS is.
 It is worth noting that from the derivation of \eqref{eq:cov00}, it can be found that $\mathbb{P}_{\mathrm {rel}}^{\mathrm{sig}}$ is similar to the CDF of $\phi_i(h_i,r_i)$, which 
 is normally distributed  according to Theorem 1.
 Therefore, when $\sigma_{\phi_i}$  changes, the variation of  $\mathbb{P}_{\mathrm {rel}}^{\mathrm{sig}}$ is similar to that of the CDF of $\phi_i(h_i,r_i)$. 
 The above analysis is also consistent with this phenomenon.
  Moreover, for an  open area  scenario  (park/stadium/square)\cite{2019The},  we can  rewrite 
$\mathbb{P}_{\mathrm {rel}}^{\mathrm{sig}}$ as follows.
\begin{Corollary} \label{Corollary1}
For  an open  area communication scenario, 
$\mathbb{P}_{\mathrm {rel}}^{\mathrm{sig}}$
can  be rewritten as follows:
\begin{align} \label{eq:rel1}    
 \mathbb{P}_{\mathrm {rel}}^{\mathrm{sig}}\approx\left(1-\frac{
\theta}{2\pi}\right)  \left(\frac{1}{2}+\frac{1}{2}\mathrm{erf}\left(\frac{p_{\mathrm{th}}-\mu_{\phi_i}}{\sqrt{2}\sigma_{\phi_i}}\right)\right)
\end{align}
\end{Corollary}
\begin{proof}
For the open  area
 scenario, such as a public park,
 buildings play a small
role. Therefore, we can assume that
 $\lambda_S\!\approx\!0$, and 
$\psi_i(h_i,r_i)\approx\!0$ by using  \eqref{eq:static1}. 
Then,   the conditional probability that the $i$-th UAV is available  is $\mathbb{P}(C_i|h_i,r_i)=1-\frac{
\theta}{2\pi}$. Since $\theta$ is   independent of
$h_i$ and $r_i$, we  have $\mathbb{P}(C_i)=\mathbb{P}(C_i|h_i,r_i)$. Substituting the $\mathbb{P}(C_i)$ into \eqref{eq:cov0}, we can obtain
 $\mathbb{P}_{\mathrm {rel}}^{\mathrm{sig}}$ as shown in \eqref{eq:rel1}.
\end{proof}

\subsection{Reliable Service Probability for Multiple UAVs  Case}
In this case,   we  assume that the typical user has more than one UAVs that can be connected, and it will immediately switch to any other available UAV when the  currently serving  link is blocked. 
We define $\mathbb{P}_{\mathrm {rel}}^{\mathrm{mul}}$ as the reliable service probability for this case, it represents the probability that at least one UAV is available with a dynamic blockage probability lower than the predefined  threshold $p_{\mathrm{th}}$.
Then, we obtain the approximate $\mathbb{P}_{\mathrm {rel}}^{\mathrm{mul}}$ as shown in Theorem \ref{Theorem_multiple}.
\begin{Theorem}  \label{Theorem_multiple}
The  reliable service probability  for the multiple UAVs case is given by 
\begin{align}  \label{eq:Theorem_multiple}
  \mathbb{P}_{\mathrm {rel}}^{\mathrm{mul}}\approx&1\!-\!\mathbb{P}_N(0)\!-\!\sum_{n=1}^{\infty} 
    \mathbb{P}_N(n)\prod_{i=1}^{n} \left(\frac{1}{2}\!-\!\frac{1}{2}\mathrm{erf}\left(\frac{p_{\mathrm{th}}\!-\!\mu_{\phi_i}}{\sqrt{2}\sigma_{\phi_i}}\right) \!\right),
\end{align}
where $\mu_{\phi_i}$ and  $\sigma_{\phi_i}$ are given in  \eqref{eq:P_B} and \eqref{eq:P(B|C)}, respectively,
  $i=1,\cdots,n$.
\end{Theorem}
\begin{proof}
In this case, $\mathbb{P}_{\mathrm {rel}}^{\mathrm{mul}}$ is expressed as shown in \eqref{eq:rel}. Therefore, we can get
\begin{align}  \label{eq:P_rel}
    \mathbb{P}_{\mathrm {rel}}^{\mathrm{mul}}=&1\!-\!\mathbb{P}_N(0)\!-\!\sum_{n=1}^{\infty} \mathbb{P}_N(n)\prod_{i=1}^{n} 
    \mathbb{P} \left(\phi_i(h_i,r_i)\!>\! p_{\mathrm{th}}\right)
        \nonumber \\
     \approx&1\!-\!\mathbb{P}_N(0)\!-\!\sum_{n=1}^{\infty} 
    \mathbb{P}_N(n)\!\prod_{i=1}^{n} \left(\frac{1}{2}\!-\!\frac{1}{2}\mathrm{erf}\left(\frac{p_{\mathrm{th}}\!-\!\mu_{\phi_i}}{\sqrt{2}\sigma_{\phi_i}}\right) \right),
\end{align}
where 
$\mathbb{P} \left(\phi_i(h_i,r_i)> p_{\mathrm{th}}\right)=1-\int_{0}^{p_{\mathrm{th}}}f_{\phi_i}(\phi_i)d{\phi_i}\approx\frac{1}{2}-\frac{1}{2}\mathrm{erf}\left(\frac{p_{\mathrm{th}}-\mu_{\phi_i}}{\sqrt{2}\sigma_{\phi_i}}\right)$ according to 
\eqref{eq:cov00}.
\end{proof}

Although Theorem \ref{Theorem_multiple} gives an approximate expression of $\mathbb{P}_{\mathrm {rel}}^{\mathrm{mul}}$, the expression   is  not tractable
 due to $n\in [0,\infty]$. 
Fortunately, since it is unrealistic to have an infinite number of UAVs in practice, 
we can assume the maximum number of UAVs is $K$.
Then, $\mathbb{P}_{\mathrm {rel}}^{\mathrm{mul}}$  can be  rewritten as:
 \begin{align}  \label{eq:P_rel2}
   \mathbb{P}_{\mathrm {rel}}^{\mathrm{mul}}\!\approx&1\!-\!\mathbb{P}_N(0)\!-\!\sum_{n=1}^{K} 
    \mathbb{P}_N(n)\prod_{i=1}^{n} \left(\frac{1}{2}\!-\!\frac{1}{2}\mathrm{erf}\left(\frac{p_{\mathrm{th}}\!-\!\mu_{\phi_i}}{\sqrt{2}\sigma_{\phi_i}}\right) \right),
\end{align}
which is solvable.
Furthermore, 
we can get the lower  bound of $\mathbb{P}_{\mathrm {rel}}^{\mathrm{mul}}$ as shown in Corollary \ref{Corollary2}.
\begin{Corollary} \label{Corollary2}
A  lower  bound  of $\mathbb{P}_{\mathrm {rel}}^{\mathrm{mul}}$  is given by 
 \begin{align}  \label{eq:Corollary2}
     \mathbb{P}_{\mathrm {rel}}^{\mathrm{mul}}\ge1\!-\!\exp\left({\!-\mathbb{P}(C_i)\lambda_T\pi{R}^{2}\left(\frac{1}{2}\!+\!\frac{1}{2}\mathrm{erf}\left(\frac{p_{\mathrm{th}} -\mu_{\phi}}{\sqrt{2}\sigma_{\phi}}\right)\right)}\right),
\end{align}
where $\mu_{\phi}=\max(\mu_{\phi_i})$ and $\sigma_{\phi}=\max(\sigma_{\phi_i})$, $i=1,\cdots,n$. 
\end{Corollary}
\begin{proof}
Since
 erf$(\cdot)$ is an increasing function and  we focus on  $p_{\mathrm{th}}> \mu_{\phi_i}$, 
 $\mathrm{erf}\left(\frac{p_{\mathrm{th}}-\mu_{\phi_i}}{\sqrt{2}\sigma_{\phi_i}}\right)$  decreases  with the increase of $\mu_{\phi_i}$ and $\sigma_{\phi_i}$, respectively. Hence, $\mathbb{P}_{\mathrm {rel}}^{\mathrm{mul}}$ in \eqref{eq:Theorem_multiple} can achieve a lower bound
when  $\mu_{\phi_1}=\cdots=\mu_{\phi_n}=\max(\mu_{\phi_i})$ and   $\sigma_{\phi_1}=\cdots=\sigma_{\phi_n}=\max(\sigma_{\phi_i})$,
$i=1,\cdots,n$, i.e.,
\begin{align}  \label{eq:P_rel1}
    \mathbb{P}_{\mathrm {rel}}^{\mathrm{mul}}&\ge1-\sum_{n=0}^{\infty} \mathbb{P}_N(n)\left(\frac{1}{2}-\frac{1}{2}\mathrm{erf}\left(\frac{p_{\mathrm{th}} -\mu_{\phi}}{\sqrt{2}\sigma_{\phi}}\right) \right)^n, \nonumber \\
    &=1\!-\!\exp\left({\!-\mathbb{P}(C_i)\lambda_T\pi{R}^{2}\left(\frac{1}{2}\!+\!\frac{1}{2}\mathrm{erf}\left(\frac{p_{\mathrm{th}} \!-\!\mu_{\phi}}{\sqrt{2}\sigma_{\phi}}\right)\right)}\right), 
\end{align}  
where the last step is obtained by substituting \eqref{eq:{P}_{N}} 
into \eqref{eq:Theorem_multiple}.
\end{proof}

According to Theorem \ref{Theorem_multiple} and Corollary \ref{Corollary2}, we  can conclude that  the larger the $\sigma$ is, the smaller the $\mathbb{P}_{\mathrm {rel}}^{\mathrm{mul}}$ is, and the lower the QoS is, which is consistent with the single UAV case. 
 Similarly, when 
 $\mu_{\phi}=\min(\mu_{\phi_i})$ and $\sigma_{\phi}=\min(\sigma_{\phi_i})$, $i=1,\cdots,n$, we can obtain  an upper  bound of $\mathbb{P}_{\mathrm {rel}}^{\mathrm{mul}}$, which 
is similar to \eqref{eq:P_rel1} and is omitted for
brevity. For an open area  scenario, 
$\mathbb{P}_{\mathrm {rel}}^{\mathrm{mul}}$  can  be obtained   by setting  $\mathbb{P}(C_i)\!=\!1\!-\!\frac{\theta}{2\pi}$,  the  reason is   consistent with the proof of  Corollary \ref{Corollary1}.

Finally, a graphical example of the impact of  $\sigma$ on the  blockage and the  QoS is shown in Fig. \ref{fig:fluctuation}.  
This example depicts two snapshots of the  links' blockage status. 
As we can  see,  the  blockage  status and reliable service in the two snapshots is different even though the blockers are the same,  which is due to the position fluctuations of the UAV. 
Moreover, as discussed in Theorem 2 and Corollary 4, the function of reliable service probability has almost similar variation characteristics to the CDF of normal distribution. Therefore, when $p_{\mathrm{th}}> \mu_{\phi_i}$, the 
 higher the UAV's fluctuation, the smaller the probability of reliable service. 
In conclusion, the larger the $\sigma$ is, the lower the
 reliable service probability is, and the worse the QoS is.
\begin{figure}
\centering
  \includegraphics[scale=0.5]{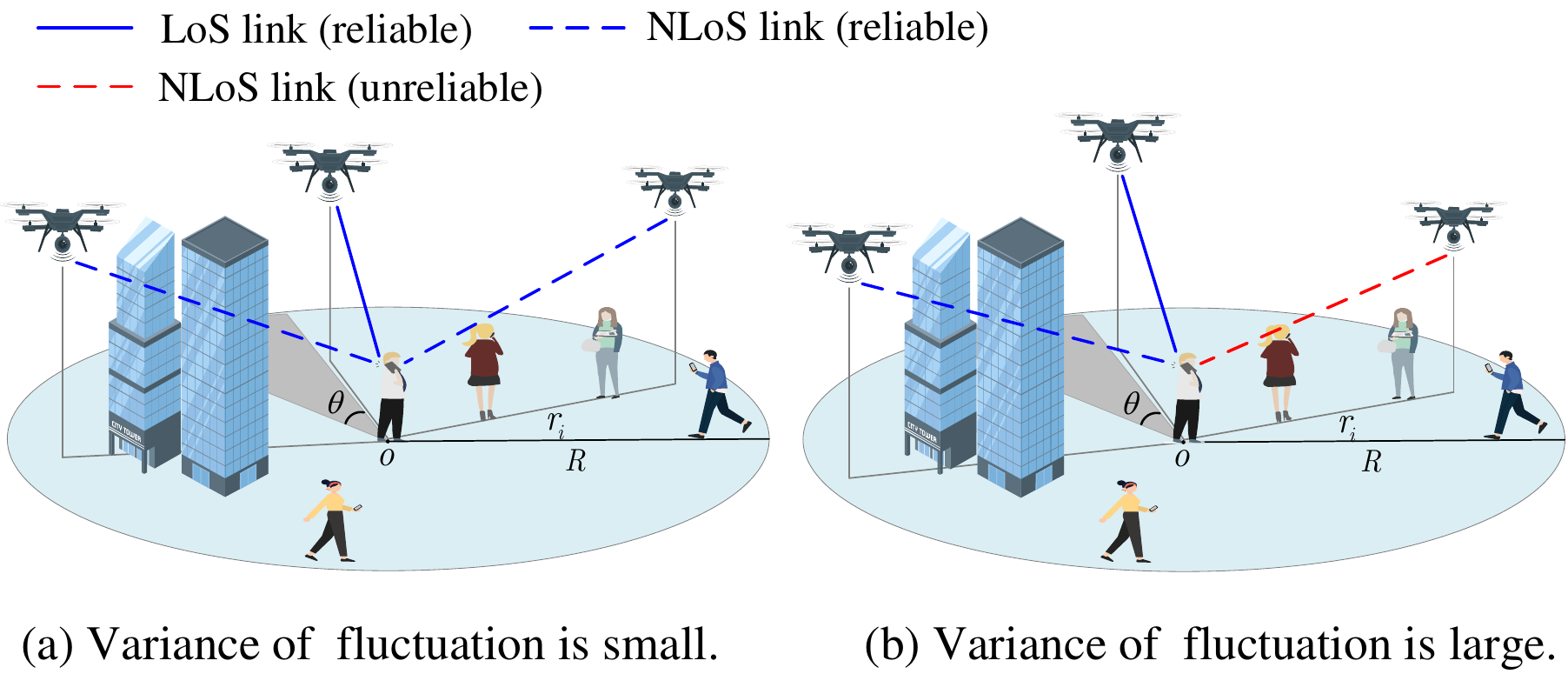}  
  \caption{ An  example of  the effect of hovering UAVs positions fluctuations on UAV-user links’ blockage status and reliable service.   
  The central positions of UAVs  in (a) and (b) are the same. 
  The snapshots in (a) and (b) show the possible positions of UAVs and blockage status under weak and strong fluctuations. 
  As can be seen, the greater the variance of position fluctuation, the more the UAV deviates from its central position. So the greater the dynamic blockage probability deviates from its mean value. Therefore, when $p_{\mathrm{th}}> \mu_{\phi_i}$, the  larger the fluctuation of UAV positions,
 the smaller the reliable service
   probability.
   }\label{fig:fluctuation}
\end{figure}
\section{Effect of Fluctuations on Coverage}
 Section \ref{Blockage}  evaluated the effect of    $\sigma$ on  reliable service probability, which is defined  respective to the blockage probability.
  This may not capture the full picture of the considered  system. Hence,  this section refers to coverage probability in terms of SNR  as another QoS measurement and evaluates
 the impact of  $\sigma$  on  this measurement. The coverage probability is defined as follows:
 
 \begin{Definition}
A user is said to be in coverage   if  at least one UAV is available with a SNR  not smaller than a predefined threshold.
The coverage probability  is denoted as $\mathbb{P}_{\mathrm {cov}}$ and
given by
 \begin{align} \label{eq:cov}
 \mathbb{P}_{\mathrm {cov}}\triangleq1\!-\!\sum_{n=0}^{\infty} \mathbb{P}_N(n)\prod_{i=0}^{n}\mathbb{P} \left(\gamma_{i} <\gamma_0 \right),
 \end{align}
where $\gamma_{i}$ denotes the SNR  at the  user from the  $i$-th UAV,
$\mathbb{P}_N(n)\prod_{i=0}^{n}\mathbb{P} \left(\gamma_{i} <\gamma_0 \right)$  denotes the probability that $n$ UAVs are available and the SNR of each link is smaller than the  
 threshold  $\gamma_0$.
Therefore, the second term of \eqref{eq:cov} denotes the probability that the user is not in coverage, and we consider that 
  $\mathbb{P} \left(\gamma_{i} <\gamma_0 \right)=1$ when $i=0$ since in this case there are no available UAVs.
  \end{Definition}

Next, we first 
develop a channel model for the considered  system. Then, 
the  SNR at the user is expressed, and the     coverage probability for single and multiple UAV cases are derived.

\subsection{Channel Model}
Since  UAV-user  link  is only dynamically blocked
 when the UAV is available, 
the LoS probability of the $i$-th  link 
can be expressed as follows:
\begin{equation}
      \label{eq:LoS}
     \mathbb{P}^{\rm{LoS}}_i=\tilde{\phi_i}(h_i,r_i)=\frac{\omega(h_i-h_R)}{{\rho r_i+\omega(h_i-h_R)}},
    \end{equation} 
 and we assume that  the channel between  UAVs and the target user is based on the dominant LoS.
Therefore,  the channel gain $g_{i}$ of the $i$-th UAV-user link is given  by  \cite{2019Accessing}:
\begin{align}  \label{eq:gain}
g_{i}=\mathbb{P}^{\rm{LoS}}_i\beta_{0}d_i^{-\alpha},
\end{align} 
where  $\beta_{0}$  is the path loss (PL) at unit distance,  $\alpha$  is    parameter
of the  PL, $d_i\!=\!\sqrt{r_i^2\!+\!(h_i\!-\!h_R)^2}$  is the 3D 
 distance from the $i$-th UAV to the typical user. Then, the  SNR  at the  user is given by
\begin{equation} \label{eq:SNR}
  \gamma_i=\frac{P_t g_i}{N_0},
\end{equation}
where $P_t$ is the transmit power of the UAV and $N_0$ is  the  noise power.  It is clear that  $ \gamma_i$ is  a
function of random variables $r_i$ and  $h_i$, which means  that 
 the position fluctuations of  UAVs will   affect the  $ \gamma_i$. Therefore, the coverage probability given  in Definition 2 can 
 capture these effects.

 
\subsection{Coverage Probability for Single UAV Case}
Similar to $\mathbb{P}_{\mathrm {rel}}^{\mathrm{sig}}$,
we define $\mathbb{P}_{\mathrm {cov}}^{\mathrm{sig}}$  as the coverage probability for the single UAV case, it represents the probability that the $i$-th UAV is available and the SNR of the $i$-th link is not smaller  than the threshold $\gamma_0$.
 Therefore, using Definition 2 and similar to  \eqref{eq:P_single} and \eqref{eq:cov00}, we first can get
\begin{align}  \label{eq:Pcov_single}   
    \mathbb{P}_{\mathrm {cov}}^{\mathrm{sig}}&=1\!-\!\sum_{n=0}^{1} \mathbb{P}_N(n)\prod_{i=0}^{n}\mathbb{P} \left(\gamma_{i} <\gamma_0 \right) \nonumber \\
    &=\mathbb{P}(C_i)\mathbb{P} \left(\gamma_{i} \ge\gamma_0 \right).
\end{align}
Then,  $\mathbb{P}_{\mathrm {cov}}^{\mathrm{sig}}$  is approximately 
 obtained as shown in Theorem \ref{Theorem2}.  \begin{Theorem}\label{Theorem2}
The   coverage probability  for the single UAV case is given by 
\begin{align}  \label{eq:P_out11}
\mathbb{P}_{\mathrm {cov}}^{\mathrm{sig}}&\approx \mathbb{P}(C_i)\left(1\!-\!Q_1\left(\frac{\mu_{d_i}}{\sigma},\frac{\tau_i}{\sigma}\right)\right),
\end{align}
where  $\mathbb{P}(C_i)$ is given in \eqref{eq:C_i_appro}, $Q_1(a,b)$ is the Marcum
Q-function and is  expressed in \eqref{eq:Q},
  $\mu_{d_i}=\sqrt{\mu_{r_i}^2+(\mu_{h_i}-h_R)^2}$, and $\tau_i={\sqrt[{\alpha}]{\frac{P_t\beta_0\omega(\mu_{h_i}\!-\!h_R)}{N_0\gamma_0({
\rho \mu_{r_i}+\omega(\mu_{h_i}\!-\!h_R)})}
}}$ is used for convenience.
\end{Theorem}
\begin{proof}
The proof is given in Appendix D.
\end{proof}
Since  Marcum Q-function is a standard function
 that is easy to compute\cite{20213DChannel},
using Theorem \ref{Theorem2}, 
   we can quickly  evaluate the QoS of   the  
considered   system
  without resorting to time-consuming simulation, especially the impact of UAV position fluctuations on the coverage probability. In addition, 
when an open area communication scenario is considered, 
the coverage  probability  can be obtained   by setting  $\mathbb{P}(C_i)\!=\!1\!-\!\frac{\theta}{2\pi}$,  the  reason is   consistent with the proof of  Corollary \ref{Corollary1}.
\subsection{Coverage Probability for Multiple UAVs Case}
Similar to $\mathbb{P}_{\mathrm {rel}}^{\mathrm{mul}}$,
 we define $\mathbb{P}_{\mathrm {cov}}^{\mathrm{mul}}$ as the coverage probability for the multiple UAVs case, it represents the probability that at least one UAV is available with a SNR not smaller  than the   threshold $\gamma_0$.
Then, we can obtain   $\mathbb{P}_{\mathrm {cov}}^{\mathrm{mul}}$ as follows:
\begin{Theorem}  \label{Theorem_cov_multiple}
The  coverage probability  for the multiple UAVs case is given by 
\begin{align}  \label{eq:Theorem_cov_multiple}
  \mathbb{P}_{\mathrm {cov}}^{\mathrm{mul}}\approx&1\!-\!\mathbb{P}_N(0)\!-\!\sum_{n=1}^{\infty} 
    \mathbb{P}_N(n)\prod_{i=1}^{n}Q_1\left(\frac{\mu_{d_i}}{\sigma},\frac{\tau_i}{\sigma}\right),
\end{align}
where  $\tau_i={\sqrt[{\alpha}]{\frac{P_t\beta_0\omega(\mu_{h_i}\!-\!h_R)}{N_0\gamma_0({
\rho \mu_{r_i}+\omega(\mu_{h_i}\!-\!h_R)})}
}}$, $i=1,\cdots,n$.
\end{Theorem}
\begin{proof}
In this case, $\mathbb{P}_{\mathrm {cov}}^{\mathrm{mul}}$ is expressed as shown in \eqref{eq:cov}. Therefore, we can get 
\begin{align}  \label{eq:P_cov_mul}
    \mathbb{P}_{\mathrm {cov}}^{\mathrm{mul}}=&1-\mathbb{P}_N(0)\!-\!\sum_{n=1}^{\infty} \mathbb{P}_N(n)\prod_{i=1}^{n}\mathbb{P} \left(\gamma_{i} <\gamma_0 \right)
        \nonumber \\
     \approx&1\!-\!\mathbb{P}_N(0)\!-\!\sum_{n=1}^{\infty} 
    \mathbb{P}_N(n)\prod_{i=1}^{n}Q_1\left(\frac{\mu_{d_i}}{\sigma},\frac{\tau_i}{\sigma}\right),
\end{align}
where  $\mathbb{P} \left(\gamma_{i} <\gamma_0 \right)\!=\!1\!-\mathbb{P} \left(\gamma_{i} \ge\gamma_0 \right)\approx Q_1\left(\frac{\mu_{d_i}}{\sigma},\frac{\tau_i}{\sigma}\right)$ with the aid of \eqref{eq:Pcov_single} and \eqref{eq:P_out11}.
\end{proof}
Similar to \eqref{eq:P_rel2}, given the limited number of UAVs in practice, $\mathbb{P}_{\mathrm {cov}}^{\mathrm{mul}}$ 
can be  rewritten as
 \begin{align}  \label{eq:P_cov_mul2}
    \mathbb{P}_{\mathrm {cov}}^{\mathrm{mul}}
     \approx&1\!-\!\mathbb{P}_N(0)\!-\!\sum_{n=1}^{K} 
\mathbb{P}_N(n)\prod_{i=1}^{n}Q_1\left(\frac{\mu_{d_i}}{\sigma},\frac{\tau_i}{\sigma}\right).
\end{align}
Furthermore,  a  lower  bound of $\mathbb{P}_{\mathrm {cov}}^{\mathrm{mul}}$ is  obtained as shown in Corollary \ref{Corollary_cov}.
\begin{Corollary} \label{Corollary_cov}
A  lower  bound  of $\mathbb{P}_{\mathrm {cov}}^{\mathrm{mul}}$  is given by 
 \begin{align}  \label{eq:Corollary_cov}
     \mathbb{P}_{\mathrm {cov}}^{\mathrm{mul}}\ge1\!-\!\exp\left({\!-\mathbb{P}(C_i)\lambda_T\pi{R}^{2}\left(1-Q_1\left(\frac{\mu_{d}}{\sigma},\frac{\tau}{\sigma}\right)\right)}\right),
\end{align}
where $\mu_d=\max(\mu_{d_i})$ and $\tau=\min(\tau_i)$, $i=1,\cdots,n$. 
\end{Corollary}
\begin{proof}
According to the characteristics of   Marcum
Q-function, we know that $Q_1(a,b)$ is an increasing function of $a$ and a decreasing function of $b$. Therefore, $\mathbb{P}_{\mathrm {cov}}^{\mathrm{mul}}$ in \eqref{eq:Theorem_cov_multiple} can achieve the lower bound
when  $\mu_{d_1}\!=\!\cdots\!=\!\mu_{d_n}=\max(\mu_{d_i})$ and   $\tau_{1}\!=\!\cdots\!=\!\tau_{n}=\min(\tau_{i})$,
$i=1,\cdots,n$, i.e.,
\begin{align} 
    \mathbb{P}_{\mathrm {cov}}^{\mathrm{mul}}&\ge1-\sum_{n=0}^{\infty} \mathbb{P}_N(n)Q_1\left(\frac{\mu_{d}}{\sigma},\frac{\tau}{\sigma}\right)^n\nonumber \\
    &=1\!-\!\exp\left({\!-\mathbb{P}(C_i)\lambda_T\pi{R}^{2}\left(1-Q_1\left(\frac{\mu_{d}}{\sigma},\frac{\tau}{\sigma}\right)\right)}\right),
\end{align}
where the last step is obtained by  substituting \eqref{eq:{P}_{N}} 
into \eqref{eq:Theorem_cov_multiple}.
\end{proof}

\section{Simulation  Results}
This section performs extensive simulations to confirm the 
theoretical results. For the simulation, we consider that  the blockers are uniformly  distributed within a circular area centered at the typical user 
with a radius of $R=100$ m.
 The movement of the human blockers is generated by a    random  way-point mobility  model\cite{2016Handover}, where human  blockers will randomly select a direction  and walk in that direction at the speed of 1 m$/$s,
 the time of moving in that direction is  uniformly distributed for [0, 60] seconds. 
The detailed simulation parameters are given in Table II.

\begin{table} \label{2}
\caption{\rule{0pt}{12pt}Simulation Parameter Values}
\renewcommand\arraystretch{1}
\begin{center}
\begin{tabular}{|l|l|}
\hline \;\;\;\;\;\;\;\;\;\;\;\;\;\;\;\;\;\;\textbf{Parameters}              &\;\;\;\textbf{\;\;\;Values}   \\
\hline Height of the  user/blockers, $h_R$/$h_B$         &1.4/1.8 m \cite{2019The} \\
\hline Density of  moving  blockers, $\lambda_B$          &0.01, 0.02 bl$/$m$^2$ \\
\hline Density of  static buildings, $\lambda_S$          &100 sbl$/$km$^2$\cite{2019The} \\
\hline Non-blocking rate, $\omega$ &$2$ bl$/$s \cite{2019The}\\
\hline  Fluctuation strength of  UAVs,   $\sigma$ &0 to 0.2 m\cite{2018Channel}\\
\hline  Transmit power of UAV, $P_t$        &$20$ dBm  \\
\hline    Noise power, $N_0$         &$-110$ dBm \cite{20213D} \\
\hline   SNR threshold, $\gamma_0$          &$3$ dB\cite{2018Flexible} \\
\hline   Size of  static buildings, $\mathbb{E}(l)\times\mathbb{E}(w)$          &$10$ m$\times10$ m\cite{2014Analysis}  \\
\hline  Path loss parameters,  $\alpha$, $\beta_{0}$ &2, $7\!\times\!10^{\!-5}$  \cite{20213D}\\
\hline
\end{tabular}\label{table:SFDs}
\end{center}
\end{table}

\subsection{Analysis of Reliable Service Probability}
\subsubsection{Single UAV case} Fig. \ref{fig:single_effect_sigma_Prel_densitys} 
 visualizes  the impact of  $\sigma$ on $\mathbb{P}_{\mathrm {rel}}^{\mathrm{sig}}$.  First, as expected, the higher the $\sigma$ is, the lower the $\mathbb{P}_{\mathrm {rel}}^{\mathrm{sig}}$ is.
Therefore, we  can get the higher the degree of UAV position fluctuations,
 the worse the system reliability.
 Then,  under the same $\sigma$, the higher the density  of blockers is, the smaller the $\mathbb{P}_{\mathrm {rel}}^{\mathrm{sig}}$ is. This is reasonable because the  fewer blockers on the UAV-user link, the smaller the blockage probability  and the higher the reliability of the link.
   Therefore,  the QoS in the scenario with static buildings is worse than in the open scenario without. 
  However, in the low $\sigma$ region ($\sigma<0.05$), even for different  $\lambda_B$, the value of $\mathbb{P}_{\mathrm {rel}}^{\mathrm{sig}}$ is the same and hardly changes with $\sigma$.
  This happens because under our simulation parameter settings, in a low $\sigma$ region, $\mathbb{P}_{\mathrm {rel}}^{\mathrm{sig}}\approx \mathbb{P}(C_i)$ and is independent of $\lambda_B$ and $\sigma$. The details are  described as: 
 Using \eqref{eq:P(B|C)}, we can get $\sigma_{\phi_i}^2$ is close to 0 when $\sigma$ is small. For example,
  when $\sigma\!=\!0.05$, $\lambda_B\!=\!0.02$ and other parameters given in 
  Table II, $\sigma_{\phi_i}^2\!\approx\!3.4\!\times\!10^{-11}$, so the fluctuation of  $\phi_i(h_i,r_i)$ is extremely small and  $\phi_i(h_i,r_i)$ is closely around its mean value $\mu_{\phi_i}$. Therefore,  the probability that $\phi_i(h_i,r_i)$ is less than the  threshold $p_{\mathrm{th}}$ is close to 1 since  $p_{\mathrm{th}}\!>\! \mu_{\phi_i}$,
 i.e.,  $\mathbb{P}(\phi_i(h_i,r_i)\!\le p_{\mathrm{th}})\!\approx\! \frac{1}{2}\!+\!\frac{1}{2}\mathrm{erf}\left(\frac{p_{\mathrm{th}}\! -\!\mu_{\phi_i}}{\sqrt{2}\sigma_{\phi_i}}\right)\!\approx\!1$, and   the value of $\mathbb{P}_{\mathrm {rel}}^{\mathrm{sig}}$ mainly depends on 
 $\mathbb{P}(C_i)$.
  
  \begin{figure}
\centering
  \includegraphics[scale=0.54]{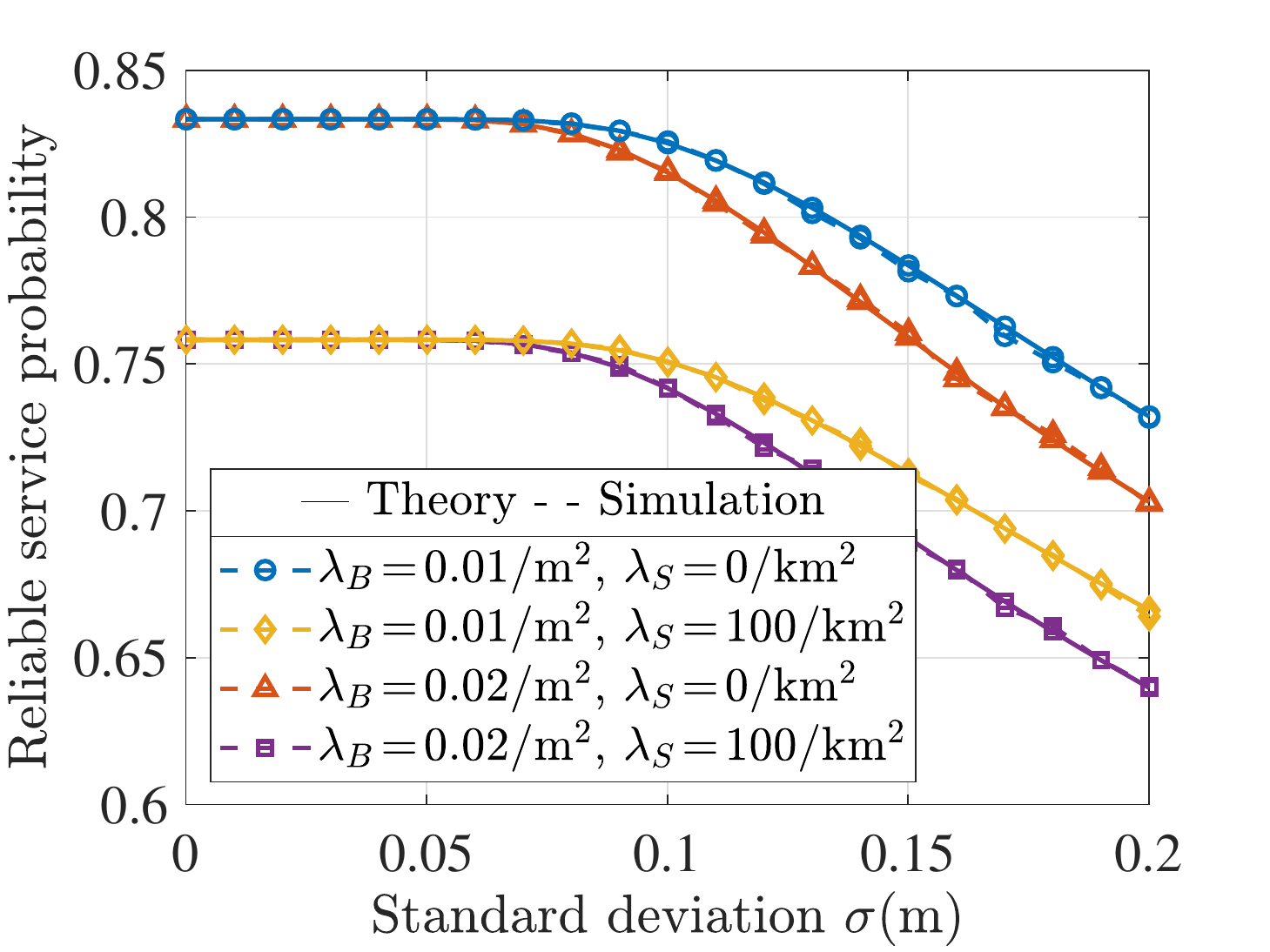 }\\ \caption{Reliable service probability $\mathbb{P}_{\mathrm {rel}}^{\mathrm{sig}}$ of the single UAV case 
    versus  UAV position fluctuation   strength $\sigma$  for  different blocker densities, where 
   $\mu_{h_i}\!=\!25$ m, $\mu_{r_i}\!=\!10$ m, $\theta\!=\!\pi/3$, and $p_{\mathrm{th}}=0.001$.
  As we can see,    the larger the $\sigma$ is, the smaller the $\mathbb{P}_{\mathrm {rel}}^{\mathrm{sig}}$ is.  }\label{fig:single_effect_sigma_Prel_densitys}
\end{figure}

 \begin{figure}
\centering
  \includegraphics[scale=0.54]{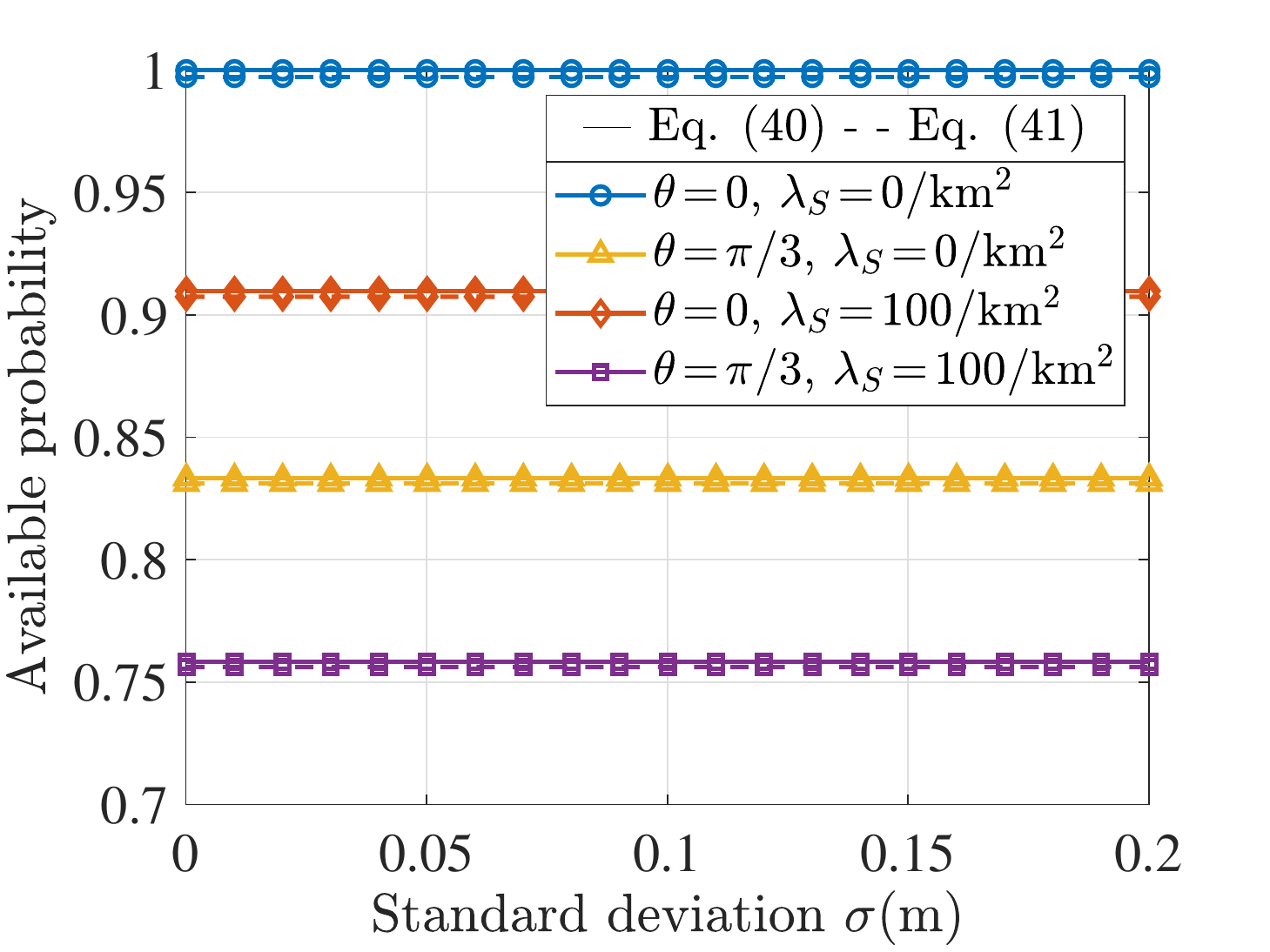}\\ \caption{Available probability $\mathbb{P}(C_i)$ versus  UAV position fluctuation   strength $\sigma$. As we can see,  the value of $\mathbb{P}(C_i)$ in \eqref{eq:C_i} and \eqref{eq:C_i_appro} is almost the same. This not only  indicates that
  $\sigma$ has little impact  on $\mathbb{P}(C_i)$, but also proves that
  the impact of $\sigma$ on the static blockage is negligible.
  Moreover,  the  curves  with $\theta=\pi/3$ are actually the $\mathbb{P}_{\mathrm {rel}}^{\mathrm{sig}}$ (when $\sigma<0.05$) 
  in Fig. \ref{fig:single_effect_sigma_Prel_densitys}.
  }\label{fig:P_ci}
\end{figure}

\begin{figure}
\centering
  \includegraphics[scale=0.54]{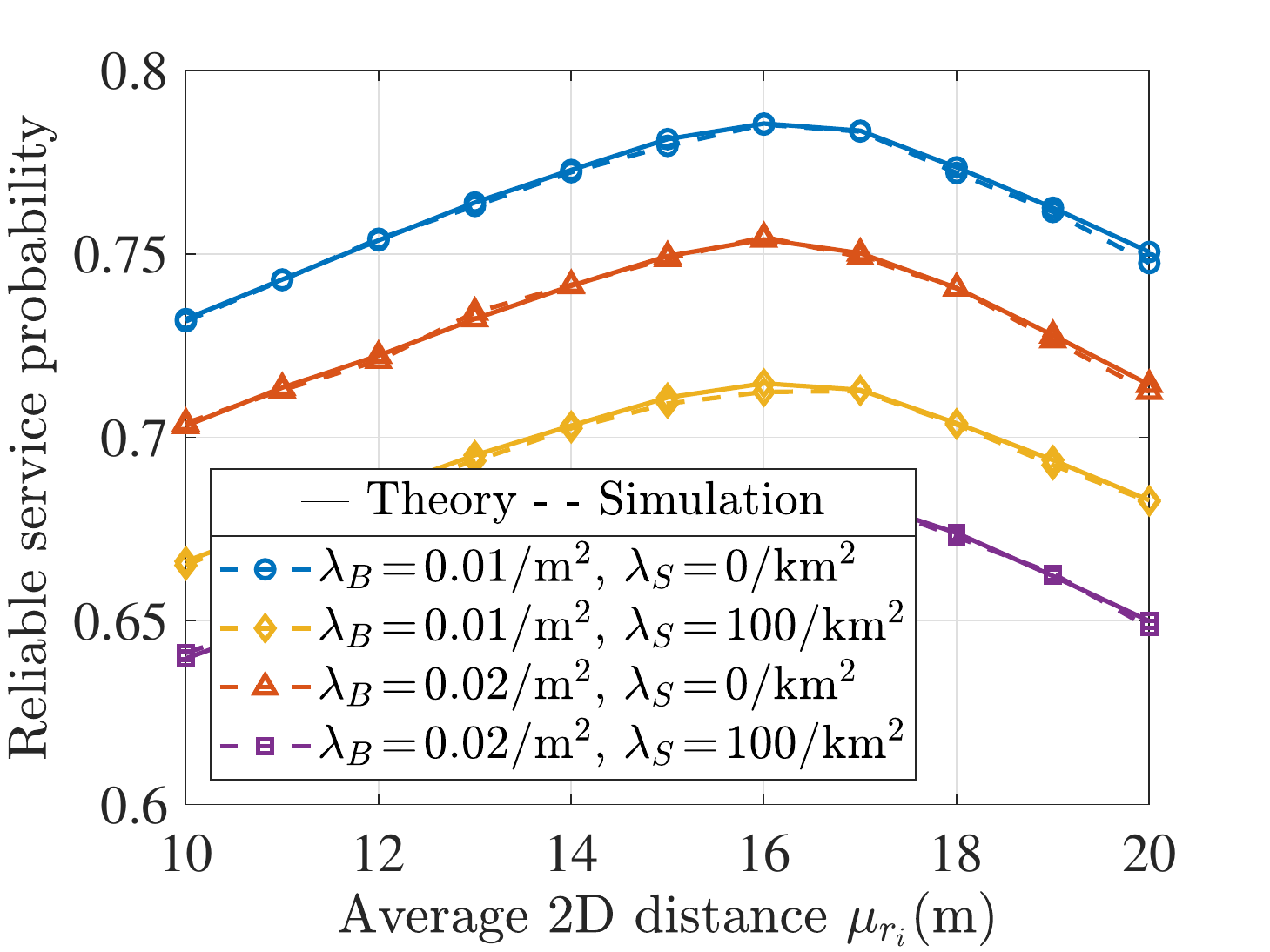}\\ \caption{
  Reliable service probability $\mathbb{P}_{\mathrm {rel}}^{\mathrm{sig}}$ of the  single UAV case
   versus average 2D  distance $\mu_{r_i}$ from the UAV to the user,  where  $\sigma\!=\!0.2$ m   and other parameters are the same as in  Fig. \ref{fig:single_effect_sigma_Prel_densitys}. As we can see, 
   there exists an optimal $\mu_{r_i}$
to  maximize  
  $\mathbb{P}_{\mathrm {rel}}^{\mathrm{sig}}$.
   }\label{fig:single_effect_sigma_Prel_densitys_higer}
\end{figure} 

 Fig.  \ref{fig:P_ci} depicts 
  the curves of $\mathbb{P}(C_i)$ given in 
  \eqref{eq:C_i} and \eqref{eq:C_i_appro}, respectively. First,
it is clear that  the impact of $\sigma$ on $\mathbb{P}(C_i)$ is negligible.
Then,  the
 curves with $\theta=\pi/3$
 are actually the $\mathbb{P}_{\mathrm {rel}}^{\mathrm{sig}}$ (when $\sigma\!<\!0.05$) 
  in Fig. \ref{fig:single_effect_sigma_Prel_densitys}. 
Therefore, at a low $\sigma$ region, $\mathbb{P}_{\mathrm {rel}}^{\mathrm{sig}}\approx \mathbb{P}(C_i)$ and $\sigma$ has little impact on it.

Fig. \ref{fig:single_effect_sigma_Prel_densitys_higer}
shows  how the average 2D  distance $\mu_{r_i}$ from the UAV to the typical user impacts $\mathbb{P}_{\mathrm {rel}}^{\mathrm{sig}}$. As we can see, there exists an optimal $\mu_{r_i}$ (about 16 m in our simulation scenario) for maximizing   $\mathbb{P}_{\mathrm {rel}}^{\mathrm{sig}}$. 
According to Corollary \ref{Corollary4} and Theorem 2,
on the one hand, we know that the greater the $\mu_{r_i}$ is, the smaller the $\sigma_{\phi_i}$ is, so the larger the  $\mathbb{P}_{\mathrm {rel}}^{\mathrm{sig}}$ is.
On the other hand, however, $\mu_{\phi_i}$ is  an increasing function of $\mu_{r_i}$, and  the greater the $\mu_{\phi_i}$ is, the smaller the $\mathbb{P}_{\mathrm {rel}}^{\mathrm{sig}}$ is. Therefore,  these two
factors  lead to the existence of an optimum $\mu_{r_i}$ for  maximizing    $\mathbb{P}_{\mathrm {rel}}^{\mathrm{sig}}$.

\begin{figure}
\centering
  \includegraphics[scale=0.54]{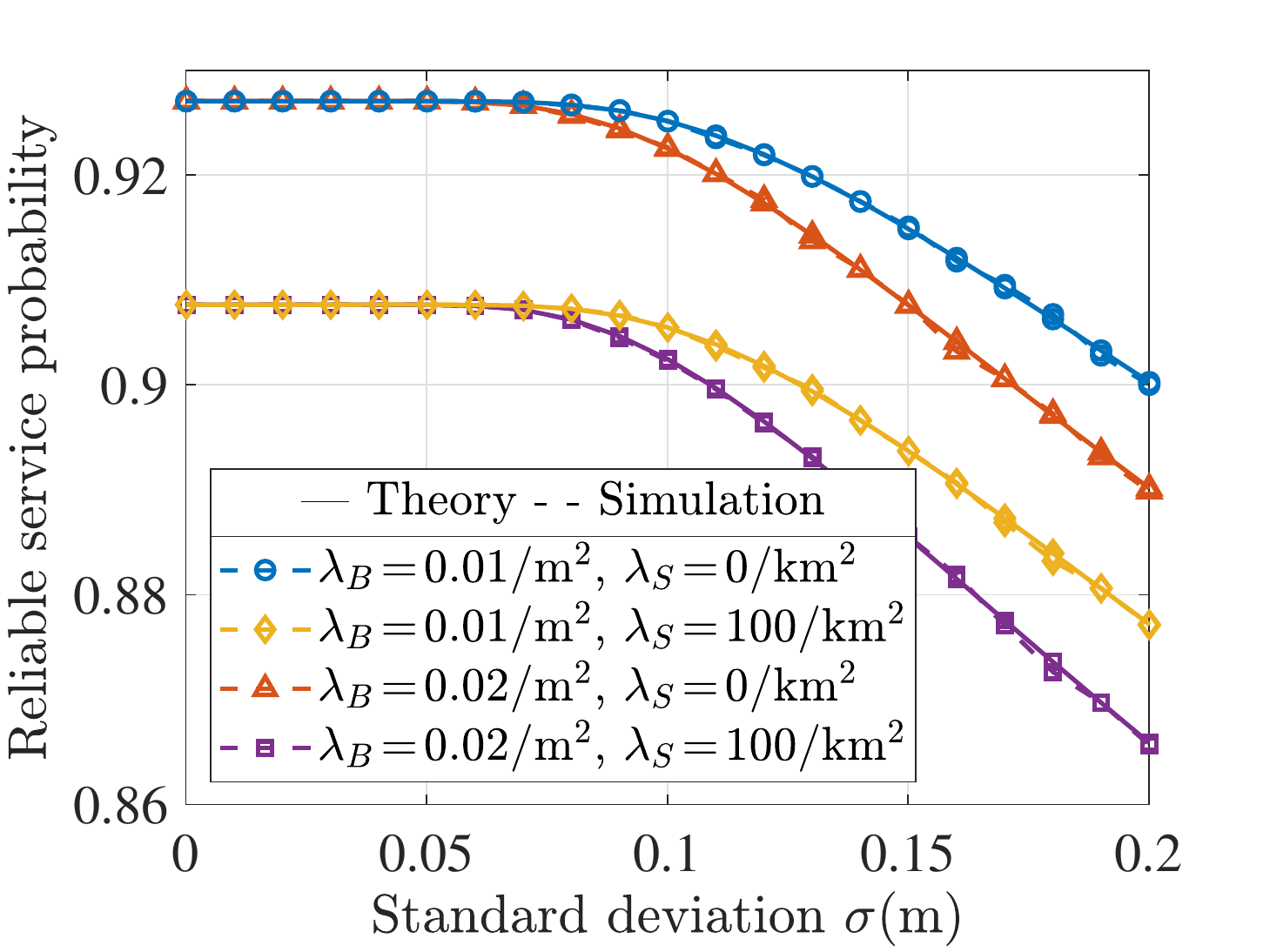}\\ \caption{Reliable service probability $\mathbb{P}_{\mathrm {rel}}^{\mathrm{mul}}$  of the multiple UAVs case 
  versus  UAV position fluctuation   strength $\sigma$  for  different blocker densities,  where  
  $\lambda_T\!=\!100/$km$^2$,
     $\mu_{h_i}\!=\!25$ m,  $\min(\mu_{r_i})\!=\!10$ m,  $\max(\mu_{r_i})\!=\!15$ m, $i=1,\cdots,n$, $\theta=\pi/3$, and $p_{\mathrm{th}}\!=\!0.001$.
  As we can see,    the larger the $\sigma$ is, the smaller the $\mathbb{P}_{\mathrm {rel}}^{\mathrm{mul}}$ is.}\label{fig:multiple_effect_sigma_Prel_densitys}
\end{figure} 

\begin{figure}
\centering
  \includegraphics[scale=0.54]{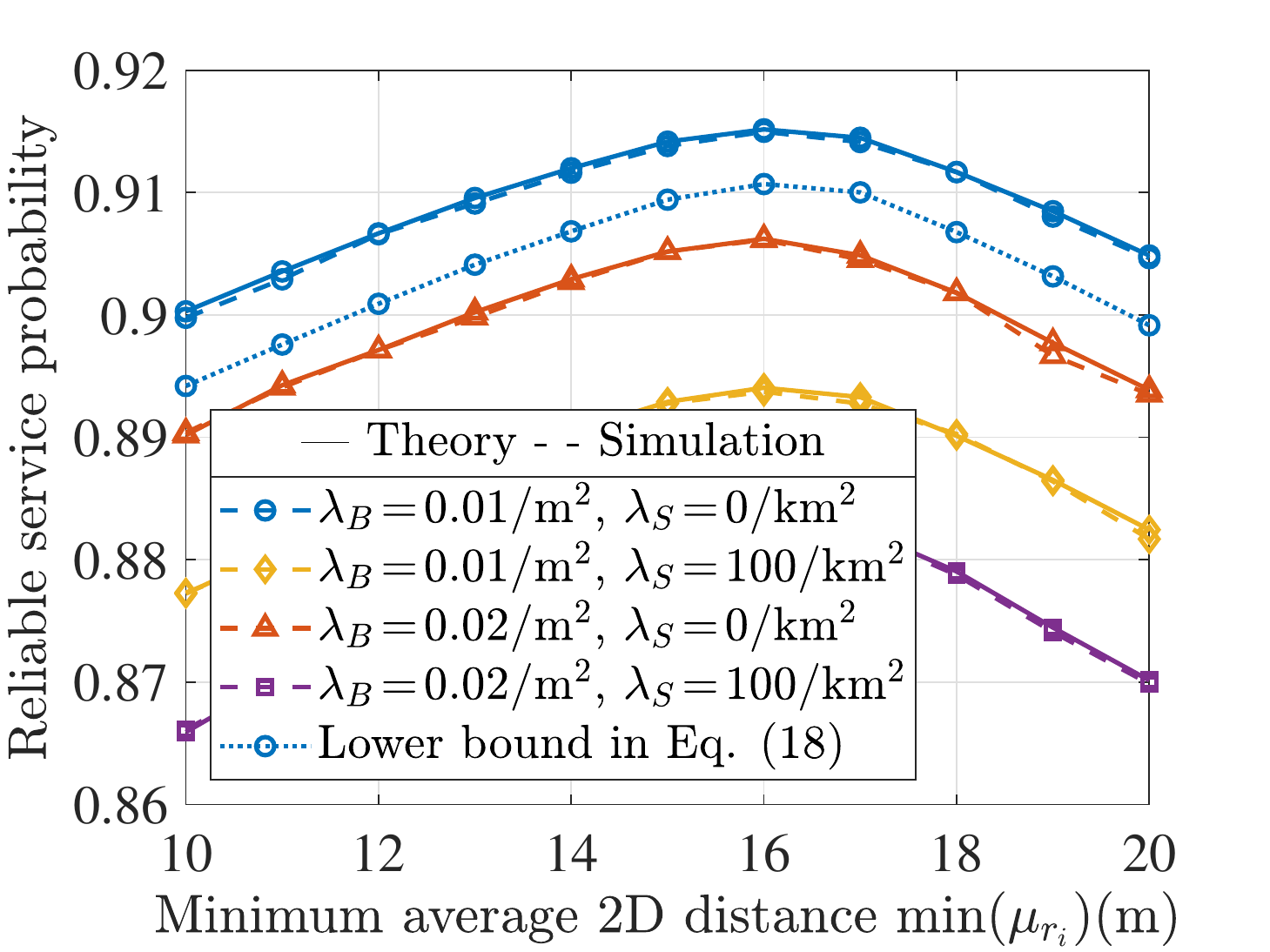}\\ \caption{Reliable service probability $\mathbb{P}_{\mathrm {rel}}^{\mathrm{mul}}$ of the multiple UAVs case  
   versus minimum average 2D  distance $\min(\mu_{r_i})$, $i\!=\!1,\cdots,n$
  from the UAV to the user  with  different blocker densities,  where  $\sigma=0.2$ m   and other parameters are the same as in  Fig. \ref{fig:multiple_effect_sigma_Prel_densitys}. As we can see, there exists an optimal $\min(\mu_{r_i})$
for maximizing  $\mathbb{P}_{\mathrm {rel}}^{\mathrm{mul}}$, and a lower bound for $\mathbb{P}_{\mathrm {rel}}^{\mathrm{mul}}$ 
in 
 \eqref{eq:Corollary2} is given.  }\label{fig:multiple_effect_sigma_Prel_densitys_higer}
\end{figure}

\subsubsection{Multiple UAVs case}
 Fig. \ref{fig:multiple_effect_sigma_Prel_densitys} and
Fig. \ref{fig:multiple_effect_sigma_Prel_densitys_higer} plot
 reliable service probability  $\mathbb{P}_{\mathrm {rel}}^{\mathrm{mul}}$  according to \eqref{eq:P_rel2}, where $K=6$.  As we can see, the larger the $\sigma$ is, the smaller the $\mathbb{P}_{\mathrm {rel}}^{\mathrm{mul}}$ is. Therefore, the influence of UAV position fluctuations  on the QoS of this case is the same as that of  the single UAV case. However, comparing these two cases,
it is obvious that $\mathbb{P}_{\mathrm {rel}}^{\mathrm{mul}}\!>\mathbb{P}_{\mathrm {rel}}^{\mathrm{sig}}$.
 This mainly benefits from the  multi-connectivity 
 strategy, i.e.,
 when the current link is blocked, the user can quickly  switch to other available UAVs to reduce the blockage. 
 Therefore,  using  multiple UAVs can   alleviate the impact of  blockage on the QoS.
 Then, similar to  the single UAV case, 
  when $\sigma$ is small ($\sigma\!<\!0.05$), we can infer $\frac{1}{2}+\frac{1}{2}\mathrm{erf}\left(\frac{p_{\mathrm{th}}\! -\!\mu_{\phi_i}}{\sqrt{2}\sigma_{\phi_i}}\right)\!\approx\!1$, so $\frac{1}{2}-\frac{1}{2}\mathrm{erf}\left(\frac{p_{\mathrm{th}}\! -\!\mu_{\phi_i}}{\sqrt{2}\sigma_{\phi_i}}\right)\!\approx\!0$,
  and  $\mathbb{P}_{\mathrm {rel}}^{\mathrm{mul}}=1\!-\!\exp\left({\!-\mathbb{P}(C_i)\lambda_T\pi{R}^{2} }\right)$.
  This indicates that in the considered  scenario,   $\mathbb{P}_{\mathrm {rel}}^{\mathrm{mul}}$
  is independent of $\lambda_B$ and $\sigma$ at a low $\sigma$ region. Noted that in this case, we assume that  all UAVs are distributed on a ring with an inner diameter of $\min(\mu_{r_i})$ an outer diameter of $\min(\mu_{r_i})\!+\!5$ m. 
 Fig. \ref{fig:multiple_effect_sigma_Prel_densitys_higer}
  studies the impact of $\min(\mu_{r_i})$ on   $\mathbb{P}_{\mathrm {rel}}^{\mathrm{mul}}$,  it is clear that an optimal $\min(\mu_{r_i})$   exists for
 maximizing $\mathbb{P}_{\mathrm {rel}}^{\mathrm{mul}}$.
 The reason is a combination of two factors, as described in the analysis of Fig. \ref{fig:single_effect_sigma_Prel_densitys_higer}.  In addition,  Fig. \ref{fig:multiple_effect_sigma_Prel_densitys_higer}
  compares $\mathbb{P}_{\mathrm {rel}}^{\mathrm{mul}}$  and a lower bound of $\mathbb{P}_{\mathrm {rel}}^{\mathrm{mul}}$  in  \eqref{eq:Corollary2}
  for $\lambda_B\!=\!0.01$ and $\lambda_S\!=\!0$. It can be observed that the value of $\mathbb{P}_{\mathrm {rel}}^{\mathrm{mul}}$   is slightly greater than the lower bound.


 \subsection{Analysis of Coverage Probability }
\subsubsection{Single UAV case}
Fig. \ref{fig:effect_sigma_Pout_densitys.eps}  studies  the effect of  $\sigma$ on the  coverage probability ${P}_{\mathrm{cov}}^{\mathrm{sig}}$ for the 
single UAV case. First, we  observe that  ${P}_{\mathrm{cov}}^{\mathrm{sig}}$ decreases as $\sigma$ increases. 
 This happens because strong fluctuations in the UAV position will result in strong fluctuations in the  SNR $\gamma_{i}$, which increases the probability that $\gamma_{i}$ will fall below the  given threshold $\gamma_0$.  Therefore,  
 ${P}_{\mathrm{cov}}^{\mathrm{sig}}$
  decreases with the increase of $\sigma$.
Then,  it is clear that in the low $\sigma$ region ($\sigma<0.05$), even for different   $\lambda_B$, the value of $\mathbb{P}_{\mathrm {cov}}^{\mathrm{sig}}$ is the same and hardly changes with $\sigma$.
This happens because
in a low $\sigma$ region, the effect of $\sigma$ on the  $\gamma_{i}$   is  negligible. Therefore,
the probability that  $\gamma_{i}$ is less than  $\gamma_0$ is close to 0, as $\gamma_{i}$ is slightly greater than $\gamma_0$  in our simulation settings, i.e.,
$\mathbb{P} \left(\gamma_{i} \ge\gamma_0 \right)\approx1$. Therefore, $\mathbb{P}_{\mathrm {cov}}^{\mathrm{sig}}\approx \mathbb{P}(C_i)$ and is independent of $\lambda_B$ and $\sigma$ by using
\eqref{eq:Pcov_single} and \eqref{eq:C_i_appro}.

\begin{figure}
\centering
  \includegraphics[scale=0.54]{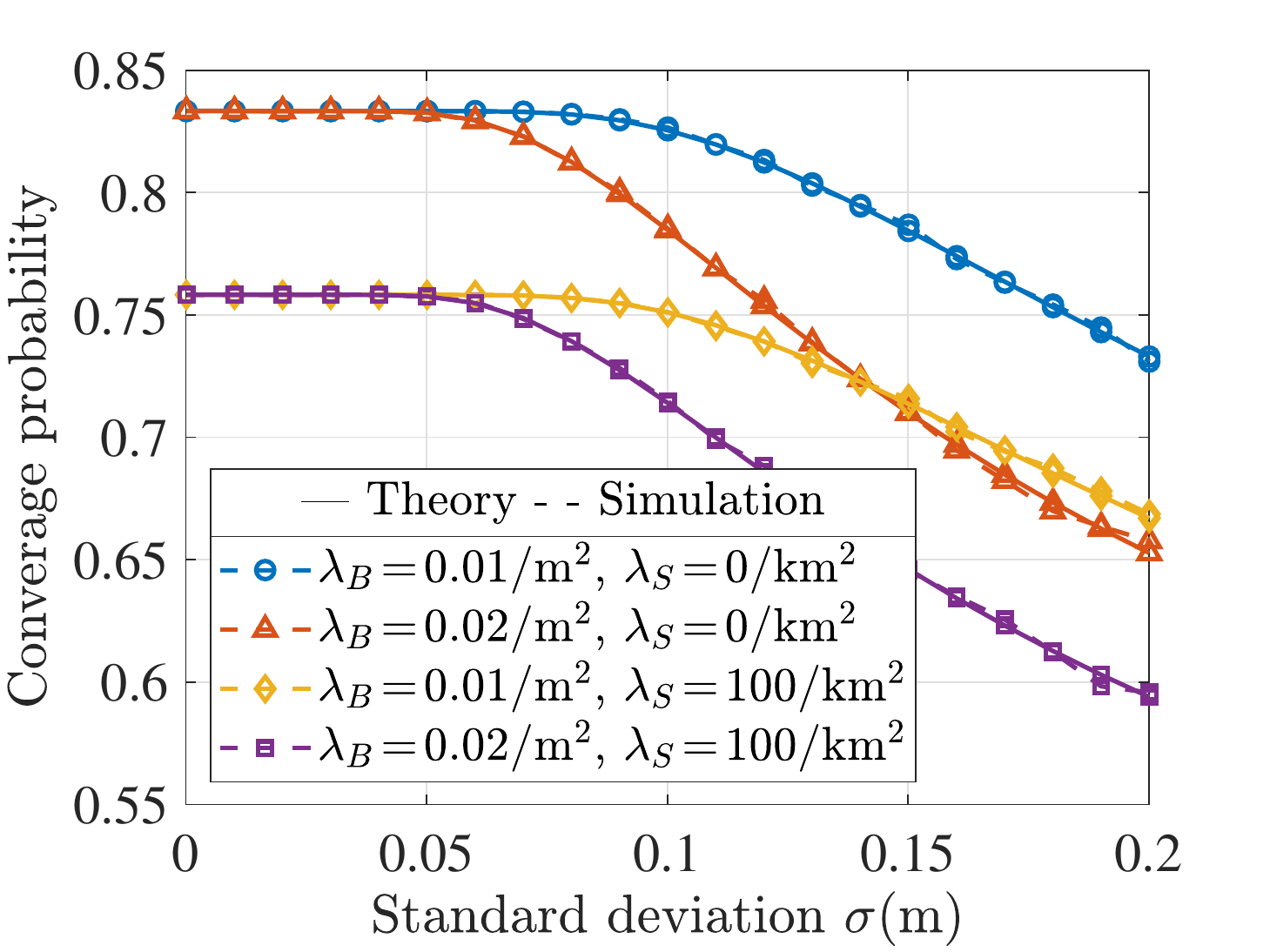}\\ \caption{
  Coverage probability $\mathbb{P}_{\mathrm {cov}}^{\mathrm{sig}}$ of the single UAV case 
  versus  UAV position fluctuation   strength $\sigma$  for  different blocker densities,  where 
    $\mu_{h_i}=25$ m, $\mu_{r_i}=50$ m, and $\theta=\pi/3$. As we can see,    the larger the $\sigma$ is, the smaller the $\mathbb{P}_{\mathrm {cov}}^{\mathrm{sig}}$ is.}\label{fig:effect_sigma_Pout_densitys.eps}
\end{figure}

\begin{figure}
\centering
  \includegraphics[scale=0.54]{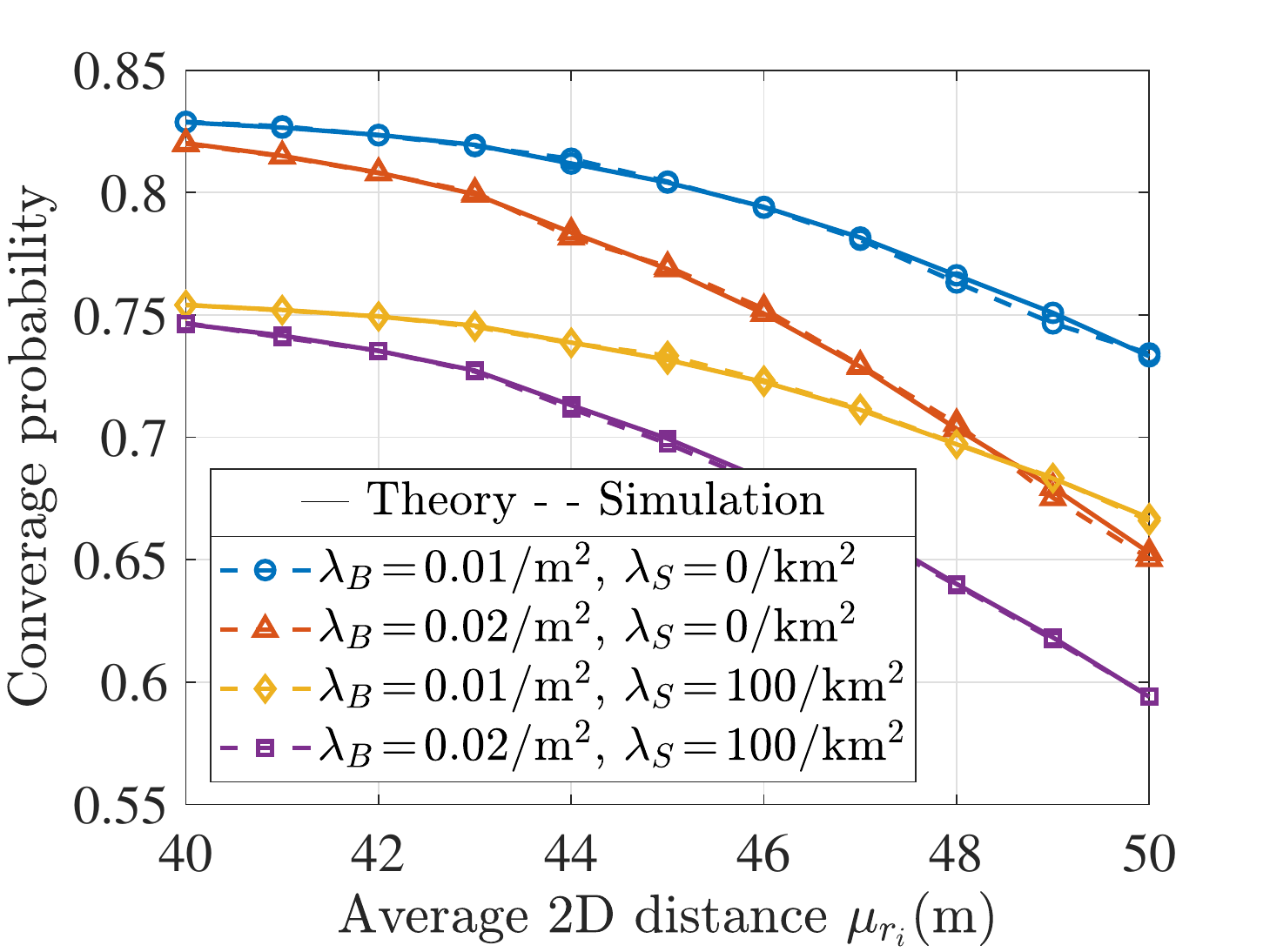}\\ \caption{
 Coverage probability $\mathbb{P}_{\mathrm {cov}}^{\mathrm{sig}}$ of the single UAV case 
  versus  average
  2D  distance $\mu_{r_i}$ from the UAV to the user,  where  $\sigma=0.2$ m   and other parameters are the same as in Fig.   \ref{fig:effect_sigma_Pout_densitys.eps}.  As we can see, the
   larger the $\mu_{r_i}$ is, the smaller the $\mathbb{P}_{\mathrm {cov}}^{\mathrm{sig}}$ is.
 }\label{fig:effect_sigma_Pout_densitys_low_r.eps}
\end{figure}

 Fig. \ref{fig:effect_sigma_Pout_densitys_low_r.eps} investigates the impact of    average 2D  distance $\mu_{r_i}$ from the UAV to the user    on   $\mathbb{P}_{\mathrm {cov}}^{\mathrm{sig}}$. It is clear that the larger the $\mu_{r_i}$  is, the smaller the $\mathbb{P}_{\mathrm {cov}}^{\mathrm{sig}}$ is. The reason is that increasing $\mu_{r_i}$  will increase the distance between transceivers and the blockage probability. 
 Therefore, 
 the SNR received by the user is reduced, and $\mathbb{P}_{\mathrm {cov}}^{\mathrm{sig}}$  is decreased  according to \eqref{eq:Pcov_single}. The result is also
 consistent with Theorem \ref{Theorem2}.
 Specifically,  $\mu_{d_i}=\sqrt{\mu_{r_i}^2+(\mu_{h_i}-h_R)^2}$ is an increasing function of  $\mu_{r_i}$, and 
 $Q_1\left(\frac{\mu_{d_i}}{\sigma},\frac{\tau_i}{\sigma}\right)$
increase with the increase of $\mu_{d_i}$
according to  the characteristics of Marcum Q-function. Therefore, $\mathbb{P}_{\mathrm {cov}}^{\mathrm{sig}}$ 
is a decreasing function of $\mu_{r_i}$.
  Additionally,
  combined with the  analysis of  reliable service probability in Fig. \ref{fig:single_effect_sigma_Prel_densitys_higer}, one  thing we found very interesting is that when $\mu_{r_i}$ is at a relatively low region, the user has a higher probability of getting reliable service and coverage. After all, selecting an appropriate $\mu_{r_i}$ can improve the service quality.
  
Fig. \ref{fig:min_Pout} shows   $\mathbb{P}_{\mathrm {cov}}^{\mathrm{sig}}$
as a function of average
  UAV height $\mu_{h_i}$. As  we expected, 
 there is  an optimal $\mu_{h_i}$ that  maximizes  $\mathbb{P}_{\mathrm {cov}}^{\mathrm{sig}}$.  Two factors cause the result:  
   On the one hand,  when the UAV is  at a low height, the LoS probability will be small. The resulting small channel gain $g_i$ increases the probability of  SNR falling below the given threshold,  thereby reducing coverage. On the other hand, for the UAV with high height, even though the link is in LoS, the increasing distance between the user and the UAV will also reduce the SNR of the link, thus leading to a low probability of coverage. Therefore, these two contributors eventually lead to the existence of  an optimum $\mu_{h_i}$ for
 maximum   $\mathbb{P}_{\mathrm {cov}}^{\mathrm{sig}}$. 

\begin{figure}
\centering
  \includegraphics[scale=0.51]{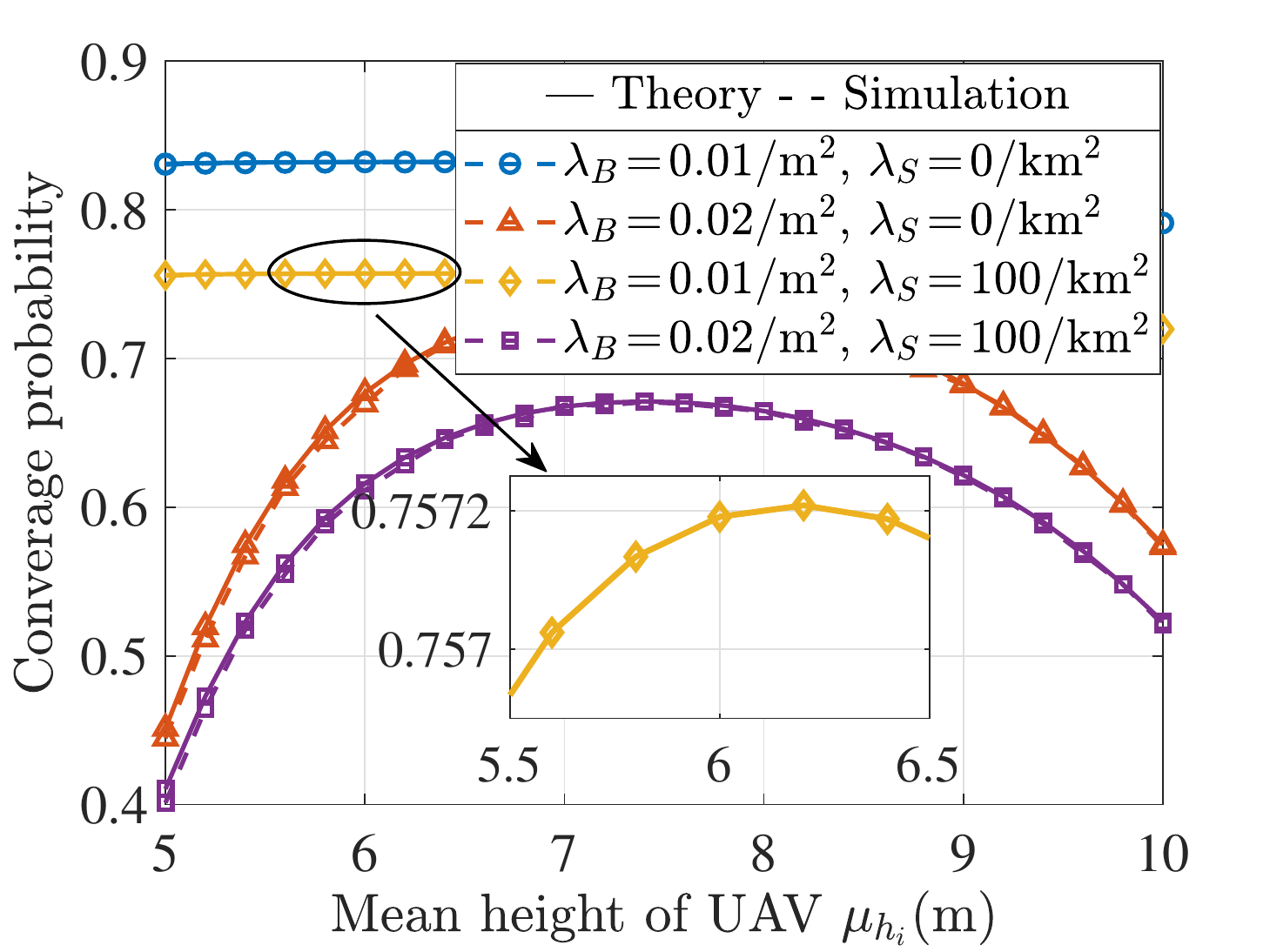}\\ \caption{
   Coverage probability $\mathbb{P}_{\mathrm {cov}}^{\mathrm{sig}}$ of the single UAV case 
  versus   mean 
  height of UAV $\mu_{h_i}$ for  different blocker densities,  where $\sigma=0.2$ m, $\theta=\pi/3$, and $\mu_{r_i}=55$ m.  As we can see, there exists an optimal $\mu_{h_i}$ that maximize $\mathbb{P}_{\mathrm {cov}}^{\mathrm{sig}}$.
 }\label{fig:min_Pout}
\end{figure}

\begin{figure}
\centering
  \includegraphics[scale=0.54]{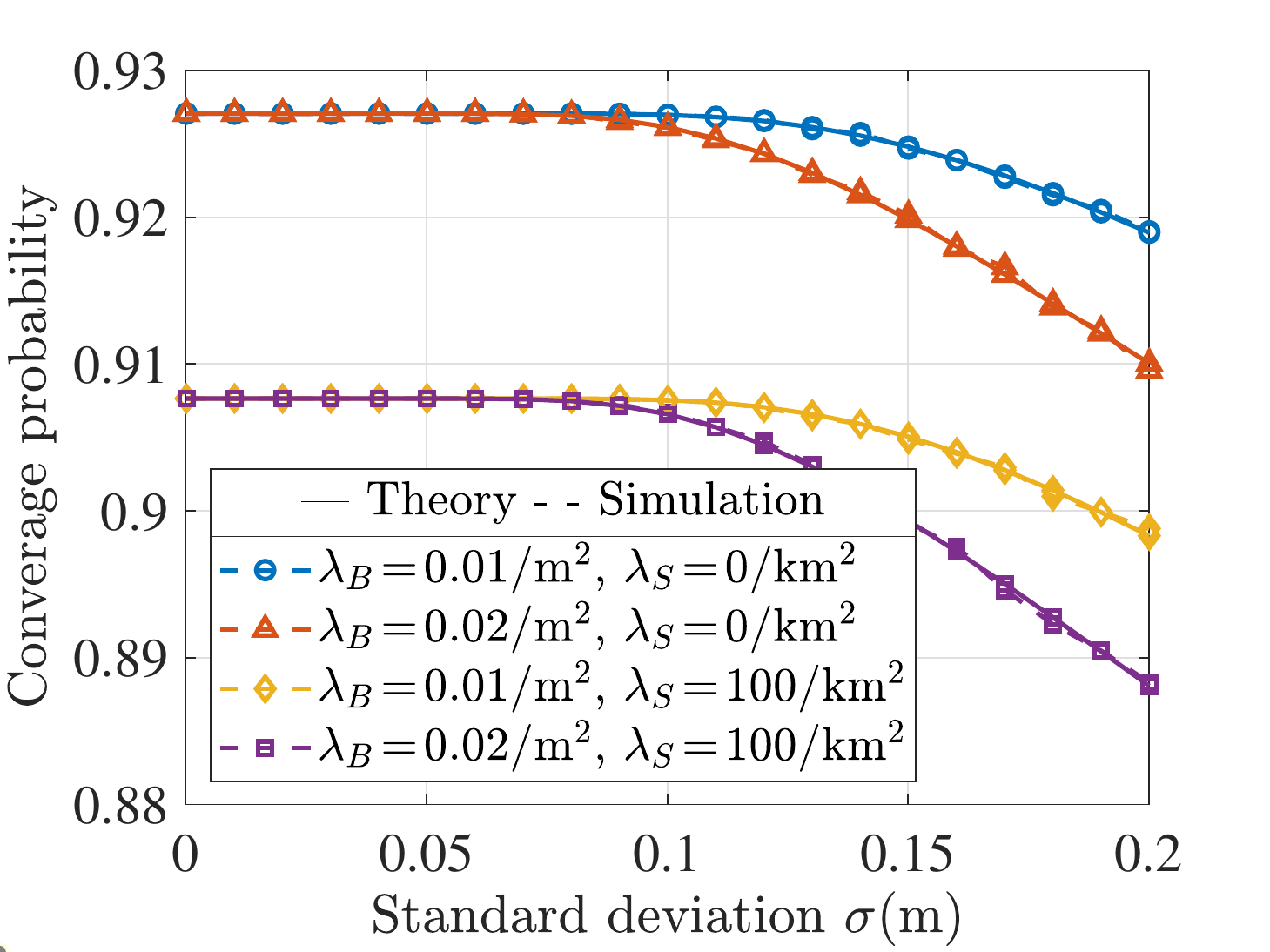}\\ \caption{Coverage probability 
  $\mathbb{P}_{\mathrm {cov}}^{\mathrm{mul}}$ of the multiple UAVs case 
  versus  UAV position fluctuation   strength $\sigma$  with  different blocker densities,  where $\lambda_T\!=\!100/$km$^2$,
    $\mu_{h_i}=25$ m, $\min(\mu_{r_i})\!=\!45$ m,  $\max(\mu_{r_i})\!=\!50$ m, $i=1,\cdots,n$,
and $\theta=\pi/3$. As we can see,    the larger the $\sigma$ is, the smaller the  $\mathbb{P}_{\mathrm {cov}}^{\mathrm{mul}}$ is.}\label{fig:multiple_effect_sigma_Pcov}
\end{figure}

\subsubsection{Multiple UAVs case}
Fig. \ref{fig:multiple_effect_sigma_Pcov}  plots  the effect of  $\sigma$ on    coverage probability $\mathbb{P}_{\mathrm {cov}}^{\mathrm{mul}}$. Similar to the single UAV case in 
Fig. \ref{fig:effect_sigma_Pout_densitys.eps},
 we can see that in a low $\sigma$ region ($\sigma\!<\!0.05$),  the  $\mathbb{P}_{\mathrm {cov}}^{\mathrm{mul}}$ for different  $\lambda_B$ is the same and hardly changes with $\sigma$. However,
when $\sigma$  is larger, the  $\mathbb{P}_{\mathrm {cov}}^{\mathrm{mul}}$ decreases with the increase of $\sigma$. Meanwhile,  when comparing Fig. \ref{fig:multiple_effect_sigma_Pcov} and Fig. \ref{fig:effect_sigma_Pout_densitys.eps}, it is obvious that $\mathbb{P}_{\mathrm {cov}}^{\mathrm{mul}}>\mathbb{P}_{\mathrm {cov}}^{\mathrm{sig}}$, the reason is that using multiple UAVs can  reduce the blockage of the link, thus improving the SNR and the coverage performance. Note that in this
case,   we assume that  all UAVs are distributed on a ring with an inner diameter of $\min(\mu_{r_i})$ an outer diameter of $\min(\mu_{r_i})+5$ m.

\begin{figure}
\centering
  \includegraphics[scale=0.54]{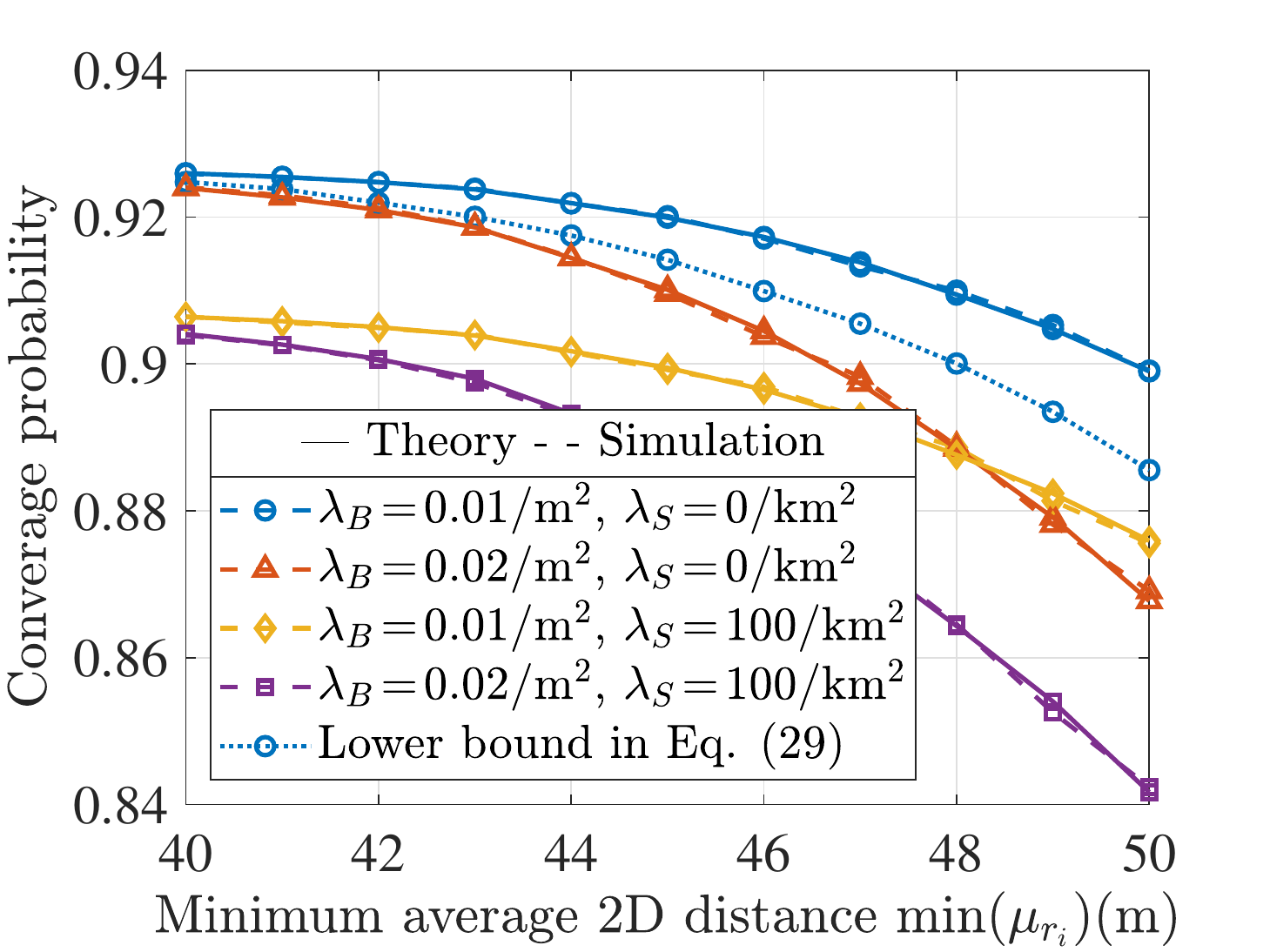}\\ \caption{Coverage probability $\mathbb{P}_{\mathrm {cov}}^{\mathrm{mul}}$ 
  of the multiple UAVs case
  versus minimum  average
  2D  distance $\min(\mu_{r_i})$,$i=1,\cdots,n$ from the UAV to the user,  where  $\sigma=0.2$ m   and other parameters are the same as in Fig.   \ref{fig:multiple_effect_sigma_Pcov}.  As we can see, the
   larger the $\min(\mu_{r_i})$ is, the smaller the $\mathbb{P}_{\mathrm {cov}}^{\mathrm{mul}}$ is, and a
   lower bound for $\mathbb{P}_{\mathrm {cov}}^{\mathrm{mul}}$
  in \eqref{eq:Corollary_cov} is given.  
 }\label{fig:multiple_min_ri_Pcov_Pcov_lower_bound}
\end{figure} 

\begin{figure}
\centering
  \includegraphics[scale=0.51]{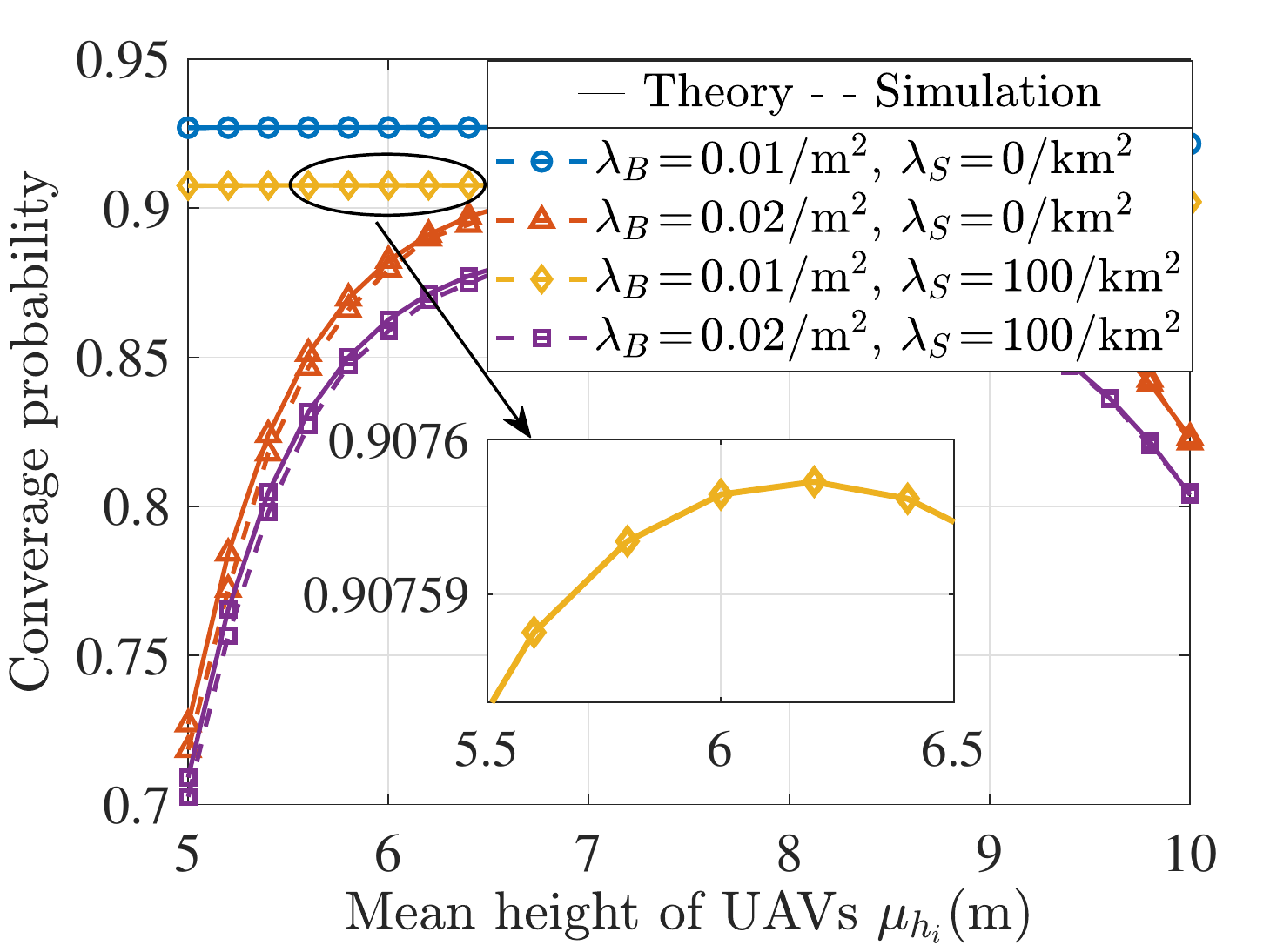}\\ \caption{
  Coverage probability  $\mathbb{P}_{\mathrm {cov}}^{\mathrm{mul}}$ of  the multiple UAVs case
  versus mean 
  height of UAVs $\mu_{h_i}$ 
  for  different blocker densities,  where $\sigma\!=\!0.2$ m, $\theta\!=\!\pi/3$, and $\min(\mu_{r_i})\!=\!55$ m.  As we can see, there exists an optimal $\mu_{h_i}$ that maximize  $\mathbb{P}_{\mathrm {cov}}^{\mathrm{mul}}$.
 }\label{fig:multiple_max_Pcov.eps}
\end{figure}

Fig. \ref{fig:multiple_min_ri_Pcov_Pcov_lower_bound}
shows the impact of $\min(\mu_{r_i})$ on 
$\mathbb{P}_{\mathrm {cov}}^{\mathrm{mul}}$. It is clear that the  $\mathbb{P}_{\mathrm {cov}}^{\mathrm{mul}}$ decrease with the increase of $\min(\mu_{r_i})$,  the reason is similar to the analysis of  Fig. \ref{fig:effect_sigma_Pout_densitys_low_r.eps} of the single UAV case, and 
  is  also consistent with the general conclusion  that the service quality of edge users is generally the worst.  In addition,   Fig.  \ref{fig:multiple_min_ri_Pcov_Pcov_lower_bound}
 compares the 
 $\mathbb{P}_{\mathrm {cov}}^{\mathrm{mul}}$ and a low bound of 
 $\mathbb{P}_{\mathrm {cov}}^{\mathrm{mul}}$ in \eqref{eq:Corollary_cov} 
 for $\lambda_B=0.01$ and $\lambda_S=0$. We  observe that the value of $\mathbb{P}_{\mathrm {rel}}^{\mathrm{mul}}$ in \eqref{eq:P_cov_mul2}  is slightly greater than the lower bound of 
 $\mathbb{P}_{\mathrm {cov}}^{\mathrm{mul}}$, which is reasonable. 
Finally,  Fig. \ref{fig:multiple_max_Pcov.eps}
 proves that there is an optimal
 $\mu_{h_i}$ for
 maximum  
 $\mathbb{P}_{\mathrm {cov}}^{\mathrm{mul}}$. The reason is  a combination of two factors, as described in
the analysis of Fig. \ref{fig:min_Pout}.

\section{Conclusion}
In this paper, we studied the critical issues affecting the QoS  of air-to-ground mmWave UAV communication systems, especially those arising from the random position fluctuations of  hovering UAVs. Accordingly, we   considered the reliable service probability respective to links blockages, and the coverage probability respective to  SNR as the 
system's key QoS measures.
 Then, we derived the 
closed-form expressions  to  evaluate the impact of UAV position fluctuations
 on these two QoS metrics.  Our  results indicated that the larger the position fluctuations of the UAV, the smaller the reliable service probability, and the smaller the coverage probability. 
Therefore,  unlike  ground mmWave communications, the QoS  of the air-to-ground mmWave UAV  communication system largely depends on the position fluctuations of the UAV. Specifically, the  fluctuation  will  cause random changes in the   blockage characteristic and the SNR of the  link, thus resulting in  unreliable communications.
Fortunately, we theoretically analyzed the effect of UAV position fluctuations on link blockages and SNR, and the relevant  results made it possible to quickly  evaluate the QoS of the  system under different levels of UAV position fluctuations. In addition, 
 according to the simulation results, we   find the  optimal 
  horizontal position and height of the UAV to  maximize the reliable service probability and  the  coverage probability, respectively, which helps establish 
  reliable  mmWave UAV communications. 
\begin{appendices}

\section{Proof of Theorem 1}
We  denote the PDF of   $\phi_i(h_i,r_i)$  as $f_{\phi_i}(\phi_i)$, and according to \eqref{eq:P_B''},
$f_{\phi_i}(\phi_i)$
can first be  expressed by
  taking the derivative of the corresponding CDF  
  $\mathbb{P} \left(\frac{\rho{r_i}}{{\rho r_i+\omega(h_i-h_R)}}\le\phi_i\right)$, i.e.,
$f_{\phi_i}(\phi_i)=\frac{\partial \mathbb{P} \left(\frac{\rho{r_i}}{{\rho r_i+\omega(h_i-h_R)}}\le\phi_i\right)} {\partial \phi_i}$. 
Therefore, it is important to analyze the CDF of  $\frac{\rho{r_i}}{{\rho r_i+\omega(h_i-h_R)}}$ in order to obtain  $f_{\phi_i}(\phi_i)$. 
Since 
$x_i$ and $y_i$ are   Gaussian distributed, $r_i\!=\!\sqrt{x_i^2+y_i^2}$
 follows a Rician distribution \cite{2018Channel}, and the PDF  can be denoted as $f_{r_i}({r_i})=\frac{r_i}{\sigma^2}{\exp{\left(-{\frac{{r_i^2}+\mu_{r_i}^2}{2{\sigma}^2}}\right)}}I_0\left(\frac{r_i \mu_{r_i}}{\sigma^2}\right)$,
 where $\mu_{r_i}\!=\!\sqrt{\mu^2_{x_i}\!+\!\mu^2_{y_i}}$, 
 $I_0(\cdot)$  denotes the modified Bessel
function of the first kind  of zero order.  In general, the 2D distance from the UAV to the user is far greater than the fluctuation strength of the UAV, i.e.,
 $\mu_{r_i}\!\gg\!\sigma$. Therefore,   $r_i$  can be
 well approximated by a Gaussian random variable
 according to \cite{sijbers1998maximum,pajevic2003parametric}, and then, the PDF of $r_i$ 
can be rewritten as follows:
\begin{align} \label{eq:Gaussian_r}
     f_{r_i}({r_i})\approx\frac{1}{\sqrt{2\pi}\sigma}{\exp\left({-{\frac{{(r_i-\mu_{r_i})}^{2}}{2{\sigma}^2}}}\right)}.
\end{align}
 Then,
 we can get 
$\rho{r_i}\!\sim\! N (\rho \mu_{r_i}, \rho^2\sigma^2)$
and $\rho{r_i}\!+\!\omega(h_i-h_R)\! \sim\! N \left(\rho \mu_{r_i}+\omega(\mu_{h_i}-h_R), (\rho^2+\omega^2)\sigma^2\right)$. Finally, 
 using the conclusions in \cite{2013existence,hinkley1969ratio},  $\frac{\rho{r_i}}{{\rho r_i+\omega(h_i-h_R)}}$ can be well modeled as a  Gaussian
distribution with   mean $\mu_{\phi_i}$ and variance $\sigma_{\phi_i}^2$. Among them,  $\mu_{\phi_i}$ and $\sigma_{\phi_i}^2$
are  given as follows:
\begin{align}
    \mu_{\phi_i}&=\frac{\rho{\mu_{r_i}}}{{\rho \mu_{r_i}\!+\!\omega(\mu_{h_i}\!-\!h_R)}},  \label{eq:P_B} \\
     \sigma_{\phi_i}^2\!&=\! \frac{ \rho^2\sigma^2}{\left(\rho \mu_{r_i}\!+\!\omega(\mu_{h_i}\!-\!h_R)\right)^2}+\frac{(\rho \mu_{r_i})^2(\rho^2+\omega^2)\sigma^2}{\left(\rho \mu_{r_i}\!+\!\omega(\mu_{h_i}\!-\!h_R)\right)^4}.  \label{eq:P(B|C)}
\end{align}
Therefore, $f_{\phi_i}(\phi_i)$ can finally be  obtained as   shown in \eqref{eq:PDF_PB}. 

\section{Proof of Corollary 1}
First, 
we take the derivative of  \eqref{eq:P_B} with respect to $\mu_{r_i}$ and   have
\begin{align}  \label{eq:sigle2}
  \frac{\partial \mu_{\phi_i}}{\partial \mu_{r_i}}=\frac{\rho\omega(\mu_{h_i}-h_R)}{\left(\rho \mu_{r_i}\!+\!\omega(\mu_{h_i}-h_R)\right)^2}>0.
 \end{align}
The reason why \eqref{eq:sigle2} is greater than 0 is that in the considered communication system, generally, the  height of the UAV  is higher than that of the user, i.e., $\mu_{h_i}>h_R$. Therefore,
 $\mu_{\phi_i}$ is  an increasing function of $\mu_{r_i}$. Next, to  obtain the  relationship between $\mu_{r_i}$ and  $\sigma_{\phi_i}$, 
we denote the first and second terms of $\sigma^2_{\phi_i}$  in \eqref{eq:P(B|C)} as $\sigma^2_{\phi_{i,1}}$  and $\sigma^2_{\phi_{i,2}}$ respectively. Then,  we can get
\begin{align} \label{eq:first term}
   \frac{\partial \sigma^2_{\phi_{i,1}}}{\partial \mu_{r_i}}\!=\!-\frac{2\rho^3\sigma^2}{\left(\rho \mu_{r_i}\!+\!\omega(\mu_{h_i}-h_R)\right)^3}<0, 
\end{align}
which indicates that $\sigma^2_{\phi_{i,1}}$   decreases with the increase of $\mu_{r_i}$. Similarly, we have
\begin{align}  \label{eq:second term}
\frac{\partial \sigma^2_{\phi_{i,2}}}{\partial \mu_{r_i}}\!=\!\frac{2\mu_{r_i}(\rho^2\!+\!\omega^2)\rho^2\sigma^2}{\left(\rho \mu_{r_i}\!+\!\omega(\mu_{h_i}-h_R)\right)^5}{\left(\rho \mu_{r_i}\!-\!\omega(\mu_{h_i}\!-\!h_R)\right)}<0,
\end{align}
which means that $\sigma^2_{\phi_{i,2}}$ is a decreasing function of $\mu_{r_i}$. 
Then, $\sigma^2_{\phi_i}$ also is a decreasing function of $\mu_{r_i}$. 
Therefore, we can conclude that   $\sigma_{\phi_i}$
 is a decreasing function of $\mu_{r_i}$.
The reason why  \eqref{eq:second term}
 is less than 0 is  as follows: In our system, 
$v$ is the speed of the moving humans,
$h_B$ and  $h_R$ are 
 the height of the moving humans and   user, respectively. Generally speaking, these parameters are fixed values in the considered communications scenario, and  $v({h_B-h_R})=0.4$ by using Table II. As such, $\rho={2{\lambda_B}v({h_B-h_R})}/{\pi}$ is only a fraction of  $\lambda_B$, and  even for $\lambda_B$ as high as $0.1$ bl/m$^2$\cite{2019The}, the value of $\rho$ is only 0.025. Therefore,  $\rho \mu_{r_i}\le\rho R =2.5$.
 However, $\mu_{h_i}$ is usually several times or even more than ten times of $h_R$
 and $\omega=2$. All in all, $\rho \mu_{r_i}\!-\!\omega(\mu_{h_i}\!-\!h_R)<0$.

\section{Proof of Lemma 1}   \label{Appendix C}
We assume that the blockage of each  link is  independent. Therefore,
given the total number of UAVs $M$, the  distribution of the number of available UAVs  $N$   can be expressed as \cite{2019The}:
\begin{align}   \label{eq:binomial}
 \mathbb{P}_{N|M}(n|m)=\binom{m}{n}\mathbb{P}(C_i)^n\left(1-\mathbb{P}(C_i)\right)^{m-n}.
\end{align}
Since $M$  is Poisson distributed,  $ \mathbb{P}_{N}(n)$   can be obtained as $\mathbb{P}_{N}(n)=\sum_{m=0}^{\infty} 
\mathbb{P}_{N|M}(n|m)\mathbb{P}_{M}(m)$ \cite[Lemma 2]{2019The}, where $\mathbb{P}_{M}(m)$ is given in Subsection \ref{UAVs}.
Therefore,  we can finally get
\begin{align}
   \mathbb{P}_{N}(n)
    &=\frac{{[ \mathbb{P}(C_i)\lambda_T\pi{R}^{2}]}^{n}}{n!}\exp\left({-\mathbb{P}(C_i)\lambda_T\pi{R}^{2}}\right),
\end{align}
where the available probability $\mathbb{P}(C_i)$ is calculated as follows:
Using  \eqref{eq:static1} and \eqref{eq:self},
the conditional probability  that the $i$-th  UAV
 is available  is $\mathbb{P}(C_i|h_i,r_i)\!=\!
 \left(1-\frac{\theta}{2\pi}\right)\exp{\left(-(\epsilon r_i+\epsilon_0)\right)}$,
and it is obviously that $ \mathbb{P}(C_i|h_i,r_i)=\mathbb{P}(C_i|r_i)$.
Therefore, to obtain  the marginal  probability $\mathbb{P}(C_i)$, we  first need to take the average of $\mathbb{P}(C_i|r_i)$ over the distribution of $r_i$, and  according to \eqref{eq:Gaussian_r},  $r_i$ follows a Gaussian distribution with mean value 
$\mu_{r_i}$.
However, $\mu_{r_i}$ is still a random variable since the center positions of the UAVs are modeled as a PPP. As a result, we  can get 
 \begin{align}  \label{eq:{P}(C_i|r_i)}
    \mathbb{P}(C_i|\mu_{r_i})&\!=\!\int_{r_i}  \mathbb{P}(C_i|h_i,r_i)f_{r_i}({r_i})d{r_i} \nonumber \\
   &\!\approx\!\left(\!1\!-\!\frac{\theta}{2\pi}\!\right)\exp{\left(\!-\epsilon \mu_{r_i}\!-\!\epsilon_0\!+\!\frac{\epsilon^2\sigma^2}{2}\right)}
   \mathrm{erf}\!\left(\frac{3\!+\!\epsilon\sigma}{\sqrt{2}}\!\right),   
\end{align}
 the
 PDF of  $\mu_{r_i}$ is $f_{\mu_{r_i}}(\mu_{r_i})=\frac{2\mu_{r_i}}{{R}^{2}}; 0<\mu_{r_i}\leq R, \forall i=1,\cdots,m.$ \cite{2019The}.
  Therefore, we  can  get
 \begin{align} \label{eq:C_i}
       \mathbb{P}(C_i)
&=\int_{\mu_{r_i}}\mathbb{P}(C_i|\mu_{r_i})f_{\mu_{r_i}}(\mu_{r_i})d{\mu_{r_i}}  \nonumber \\
&\approx\left(1\!-\!\frac{\theta}{2\pi}\right)\mathrm{erf}\left(\frac{3\!+\!\epsilon\sigma}{\sqrt{2}}\right)\frac{2\exp{(\!-\epsilon_0\!+\!\frac{\epsilon^2\sigma^2}{2})}}{R^2\epsilon^2} 
\left(1\!-\!(1\!+\!R\epsilon )\exp{(\!-R\epsilon )}\right).
\end{align}
For further insight, 
we   further approximate $\mathbb{P}(C_i)$ as follows: 
We know that
 $\mathrm{erf}\left(\frac{3}{\sqrt{2}}\right)\approx1$, and $3+\epsilon\sigma\ge3$ due to 
 $\epsilon\!=\!\frac{2} {\pi}\lambda_S(\mathbb{E}(l)\!+\!\mathbb{E}(w))\!>\!0$. Therefore,
  $ \mathrm{erf}\left(\frac{3+\epsilon\sigma}{\sqrt{2}}\right)\!\approx\!1$ according to the characteristics of error function. Then,
 since the size of  the buildings are $\mathbb{E}(l)\times\mathbb{E}(w)$, the maximum value of density of the buildings is $\lambda_S^{\max}\!=\!\frac{1}{{\mathbb{E}(l)}\mathbb{E}(w)}$, and    the 
maximum value of $\epsilon$ is given by
$\epsilon_{\max}=\frac{2} {\pi}\lambda_S^{\max}(\mathbb{E}(l)+\mathbb{E}(w))=\frac{2}{\pi}\left(\frac{1}{\mathbb{E}(w)}+\frac{1}{\mathbb{E}(l)}\right) < \frac{4}{\pi}\approx1$,
where the   inequality
is obtained based on the fact that the value of  $\mathbb{E}(w)$ and $\mathbb{E}(l)$ are  larger than 1 in practice, and the result indicates that $\epsilon\!<\!1$,
 i.e.,  $\epsilon\!>\! \epsilon^2$. We have
$\epsilon_0\!=\!\lambda_S\mathbb{E}(l)\mathbb{E}(w)\!>\!0$, 
 $\mu_{r_i}\!\gg\!\sigma$ and $\sigma<1$ in practice. Hence, $\epsilon \mu_{r_i}\!+\!\epsilon_0\!\gg\!\epsilon^2\sigma^2$, and  the value of 
 $-(\epsilon \mu_{r_i}\!+\!\epsilon_0)\!+\!\frac{\epsilon^2\sigma^2}{2}$  mainly depends on $-(\epsilon \mu_{r_i}\!+\!\epsilon_0)$, i.e.,  $\exp{\left(-\epsilon \mu_{r_i}\!-\!\epsilon_0\!+\!\frac{\epsilon^2\sigma^2}{2}\right)}\!\approx\!\exp{\left(-(\epsilon \mu_{r_i}\!+\!\epsilon_0\right))}$, 
 which verifies our  insights in  Remark 3. Therefore we can assume that $\mathbb{P}(C_i|\mu_{r_i})\approx \left(1\!-\!\frac{\theta}{2\pi}\right)\exp{\left(\!-\!(\epsilon \mu_{r_i}\!+\!\epsilon_0)\right)}$, and then,
 \begin{align} \label{eq:C_i_appro}
       \mathbb{P}(C_i)
&\!\approx\!\left(1\!-\!\frac{\theta}{2\pi}\right)\frac{2\exp{(\!-\epsilon_0)}}{R^2\epsilon^2}\left(1\!-\!(1\!+\!R\epsilon )\exp{(\!-R\epsilon )}\right),
\end{align}
Fig. \ref{fig:P_ci}
 further illustrates the accuracy of this approximation.

\section{Proof of Theorem 4}
According to \eqref{eq:Pcov_single}, we can get
$\mathbb{P}_{\mathrm {cov}}^{\mathrm{sig}}=\mathbb{P}(C_i)\mathbb{P} \left(\gamma_{i} \ge\gamma_0 \right)$,
where $\mathbb{P}(C_i)$ is given in \eqref{eq:C_i_appro}. Hence, we only need to calculate $\mathbb{P} \left(\gamma_{i} \ge\gamma_0 \right)$.
With the aid of   \eqref{eq:LoS}, \eqref{eq:gain} and
\eqref{eq:SNR}, we can get
\begin{align} \label{eq:P_out}
  \mathbb{P} \left(\gamma_{i} \ge\gamma_0 \right)=\mathbb{P}\left({\beta_0 \mathbb{P}^{\rm{LoS}}_i}{d_i^{-\alpha}}\ge\frac{N_0}{P_t}\gamma_0\right)
\end{align}
which is
related to the CDF of
${\beta_0 \mathbb{P}^{\rm{LoS}}_i}{d_i^{-\alpha}}$. Therefore, to obtain  $\mathbb{P}_{\mathrm {cov}}^{\mathrm{sig}}$, it is necessary  
 to calculate the CDF mentioned above
  by considering  random variables $h_{i}$, $r_{i}$, and $d_{i}$.
 To this end, we  define  an indicator random variable  $Z$ as follows:
 \begin{align}  \label{eq:Z}
 Z={\beta_0 \mathbb{P}^{\rm{LoS}}_i}{d_i^{-\alpha}},
 \end{align}
 and analyze the CDF of $Z$.  
First, 
we have $h_{i}\!\sim\! N (\mu_{h_i}, \sigma^2)$ and $r_{i}\!\sim\! N (\mu_{r_i}, \sigma^2)$.
 Therefore,  
  $d_{i}\!=\!\sqrt{r_{i}^2\!+\!(h_i\!-\!h_R)^2}$ is  Rician distributed, and the PDF is  $f_{d_{i}}({d_{i}})\!=\!\frac{d_{i}}{\sigma^2}{\exp{\left(\!-{\frac{{d_{i}^2}\!+\!\mu_{d_i}^2}{2{\sigma}^2}}\right)}}I_0\left(\frac{d_{i}\mu_{d_i}}{\sigma^2}\right)$, 
 where $\mu_{d_i}\!=\!\sqrt{\mu_{r_i}^2+(\mu_{h_i}\!-\!h_R)^2}$. Then, 
using  \eqref{eq:LoS} and \eqref{eq:P_B''}, we have $\mathbb{P}^{\rm{LoS}}_i\!=1-{\phi_i}(h_i,r_i)$, and $\phi_i(h_i,r_i)$ follows a normal distribution as shown in Theorem \ref{Theorem1}. Therefore, 
$\beta_0\mathbb{P}^{\rm{LoS}}_i$
also follows a normal distribution with mean $\beta_0(1-
   \mu_{\phi_i})$ and
variance $\beta_0^2\sigma_{\phi_i}^2$ where 
$\mu_{\phi_i}$ and $\sigma_{\phi_i}^2$
are given in \eqref{eq:P_B}
 and
 \eqref{eq:P(B|C)}, respectively.
Next, we will prove that the randomness of $Z$ mainly depends on $d_{i}$. In this case, the task of calculating the CDF of $Z$ can be converted to calculating the CDF of $d_{i}$. We first prove that 
the random variable $\beta_0\mathbb{P}^{\rm{LoS}}_i$  can be approximated by its mean   
$\beta_0(1-
   \mu_{\phi_i})$.
Using the proof of  Corollary \ref{Corollary4}, we  know 
  that $\rho\le0.025$, 
$\rho\mu_{r_i}\le 2.5$, and
$\rho$
is several
orders of magnitude smaller than that of $\omega(\mu_{h_i}\!-\!h_R)$ for our system setup. Therefore, we first can get 
$\rho^2\sigma^2 \ll \left(\omega(\mu_{h_i}\!-\!h_R)\right)^2$ since $\sigma<1$ in practice. Then, we have
  $\frac{ \rho^2\sigma^2}{\left(\rho \mu_{r_i}\!+\!\omega(\mu_{h_i}\!-\!h_R)\right)^2}\!\ll\!\frac{\omega(\mu_{h_i}\!-\!h_R)(\rho \mu_{r_i}\!+\!\omega(\mu_{h_i}\!-\!h_R))}{{\left(\rho \mu_{r_i}\!+\!\omega(\mu_{h_i}\!-\!h_R)\right)^2}}$. Moreover, we can get $(\rho \mu_{r_i})^2(\rho^2+\omega^2)\sigma^2\le2.5^2(0.025^2+4)0.04\approx1$  with the aid of  Table II. Therefore,  $\frac{(\rho \mu_{r_i})^2(\rho^2+\omega^2)\sigma^2}{\left(\rho \mu_{r_i}\!+\!\omega(\mu_{h_i}\!-\!h_R)\right)^4}\ll\frac{\omega(\mu_{h_i}\!-\!h_R)(\rho \mu_{r_i}\!+\!\omega(\mu_{h_i}\!-\!h_R))^3}{{\left(\rho \mu_{r_i}\!+\!\omega(\mu_{h_i}\!-\!h_R)\right)^4}}$ due to $\mu_{h_i}$ is usually several times or even more than ten times of $h_R$
and $\omega=2$. In addition, we have
$\beta_{0}=7\!\times\!10^{\!-5}$\cite{20213D}. 
Based on the above discussion, we can conclude that $\beta_0^2\left(\frac{ \rho^2\sigma^2}{\left(\rho \mu_{r_i}\!+\!\omega(\mu_{h_i}\!-\!h_R)\right)^2}+\frac{(\rho \mu_{r_i})^2(\rho^2+\omega^2)\sigma^2}{\left(\rho \mu_{r_i}\!+\!\omega(\mu_{h_i}\!-\!h_R)\right)^4}\right)\ll\frac{\beta_0\omega(\mu_{h_i}\!-\!h_R)}{{\rho \mu_{r_i}\!+\!\omega(\mu_{h_i}\!-\!h_R)}}$, i.e., $\beta_0^2\sigma_{\phi_i}^2\ll\beta_0(1-
   \mu_{\phi_i})$. Therefore, $\beta_0^2\sigma_{\phi_i}^2$ 
has little impact on $\beta_0\mathbb{P}^{\rm{LoS}}_i$.
As a result,  we can assume that $\beta_0\mathbb{P}^{\rm{LoS}}_i\approx\beta_0(1-
   \mu_{\phi_i})$, and then, substituting \eqref{eq:P_B} into \eqref{eq:Z},
$Z$ can be rewritten as
 \begin{align} \label{eq:Z1}
 Z\approx\frac{\beta_0\omega(\mu_{h_i}-h_R)}{d_i^{\alpha}\left(\rho \mu_{r_i}+\omega(\mu_{h_i}-h_R)\right)},
 \end{align}
 and  the  objective of calculating the CDF of $Z$ can be simplified as calculating the CDF of $d_{i}$.
Let us denote the CDF of $d_{i}$ as $F_{d_{i}}(d)$. Since $d_{i}$ is Rician distributed,
$F_{d_{i}}(d)$
can be  expressed in terms of  the  Marcum Q-function  as follows \cite{kam2006new,2002Closed}: 
\begin{align}  \label{eq:CDF}
    F_{d_{i}}(d)=1-Q_1\left(\frac{\mu_{d_{i}}}{\sigma},\frac{d}{\sigma}\right)
\end{align}
where  $Q_1(a,b)$ is the Marcum Q-function and is defined as
\begin{align}  \label{eq:Q}
Q_1(a,b)=\int_{b}^{\infty}x\exp{\left(-\frac{a^2+x^2}{2}\right)}I_0(ax)d{x}.
\end{align}
Finally, with the aid of \eqref{eq:Pcov_single}, \eqref{eq:P_out}, \eqref{eq:Z}, \eqref{eq:Z1} and \eqref{eq:CDF}, we have
\begin{align}  \label{eq:closed}
    \mathbb{P}_{\mathrm {cov}}^{\mathrm{sig}}
    &\approx\mathbb{P}(C_i)\mathbb{P}\left(\frac{\beta_0\omega(\mu_{h_i}-h_R)}{d_i^{\alpha}\left(\rho \mu_{r_i}+\omega(\mu_{h_i}-h_R)\right)} \!\ge\!\frac{N_0}{P_t}\gamma_0\right) \nonumber \\
&=\mathbb{P}(C_i)\mathbb{P}\left(d_{i} \le\sqrt[{\alpha}]{\frac{P_t\beta_0\omega(\mu_{h_i}\!-\!h_R)}{N_0\gamma_0({
\rho \mu_{r_i}+\omega(\mu_{h_i}\!-\!h_R)})}
}\right) \nonumber \\
    &=\mathbb{P}(C_i)\left(1\!-\!Q_1\left(\frac{\mu_{d_i}}{\sigma},\frac{\sqrt[{\alpha}]{\frac{P_t\beta_0\omega(\mu_{h_i}\!-\!h_R)}{N_0\gamma_0({
\rho \mu_{r_i}+\omega(\mu_{h_i}\!-\!h_R)})}
}}{\sigma}\right)\right),
\end{align}

\end{appendices}

\bibliographystyle{IEEEtran}
\bibliography{IEEEabrv,UAV}
\end{document}